\documentclass[smallcondensed]{svjour3}
\usepackage{amsmath, bm, bbm}
\usepackage{amssymb}
\usepackage{amsfonts}
\usepackage{graphicx}
\usepackage{float}
\usepackage{mathtools}
\usepackage{appendix}
\DeclareMathOperator\erf{erf}

\newtheorem{thm}{Theorem}[section]

\newtheorem{prop}[thm]{Proposition}

\newtheorem{define}[thm]{Definition}

\DeclarePairedDelimiter{\ceil}{\lceil}{\rceil}

\title{Partially phase-locked solutions to the Kuramoto model.}
\author{
  Jared C. Bronski \and Lan Wang
}
\institute{J. C. Bronski \at Department of Mathematics \\ University of Illinois \\ 1409 W. Green St. Urbana, IL 61801 \\ Tel.: +1-217-244-8218 \\
  \email{bronski@illinois.edu} \and
Lan Wang \at Department of Mathematics \\ University of Illinois \\ 1409 W. Green St. Urbana, IL 61801 \\ Tel.: +1-217-244-8218 \\
  \email{lanwang2@illinois.edu}
}

\date{May 14, 2020}

\begin{document}

\maketitle
\begin{abstract}
 The Kuramoto model is a canonical model for understanding
 phase-locking phenomenon.
 It is well-understood that, in the usual mean-field scaling, full phase-locking is unlikely and that it is partially phase-locked states that are important in
 applications.  Despite this, while there has been much attention
 given to existence and stability of fully phase-locked states in the
 finite $N$ Kuramoto model,
 the partially phase-locked states have
 received much less attention. In this paper we present two
 related results. Firstly, we derive an analytical criteria that, for
 sufficiently strong coupling, guarantees the existence of a 
partially phase-locked state by proving the existence of an attracting
ball around a fixed point of a subset of the oscillators. We also
derive
a larger invariant ball such that any point in it will asymptotically
converge to the attracting ball. Secondly, we consider the large $N$ (thermodynamic)
limit for the Kuramoto system with randomly distributed frequencies. Using some results of De Smet and Aeyels
on partial entrainment, we derive a  deterministic condition giving almost sure existence of a partially
entrained state for sufficiently strong coupling when the natural
frequencies of the individual oscillators are independent identically
distributed random variables, as well as upper and lower bounds on the size of the largest cluster of partially entrained oscillators. 
Interestingly in a series on numerical experiments we find that the observed size of the largest entrained cluster is predicted extremely well by 
the upper bound. 
\end{abstract}

\section{Introduction} 

\subsection{Background}

Synchronization and phase-locking phenomena are ubiquitous in the
natural world. Dynamical systems modeling a diverse collection of
phenomena including  neural signaling  \cite{Gray1994}, the beating of
the heart \cite{Torre1976} and the signaling of fire-flies 
\cite{Pikovsky2003} exhibit synchonization and phase-locking
behaviors. The (finite $N$)  Kuramoto model \cite{Kuramoto1984,Kuramoto1991}
\begin{equation} 
    \dot{\theta_i} = \omega_i - \frac{\gamma}{N} \sum_{j}
    \sin(\theta_i - \theta_j) \qquad   i = 1,2,...,N 
\label{eqn:Kuramoto}
\end{equation}
has proven to be a popular model for describing the dynamics of these
systems. Here $\theta_i \in \mathbb{T}^1 = (-\pi,\pi]$ is a phase
variable describing the state of the $i^{th}$ oscillator, $\omega_i
\in \mathbb{R}$ is the natural frequency of the $i^{th}$ oscillator,
and $\gamma>0$ is the coupling strength among the oscillators.
Here we are assuming the simplest graph topology: the case of all-to-all
coupling (complete graph) with homogeneous interactions.   
A great deal of work has been directed towards studying  necessary and/or
sufficient conditions on the critical coupling strength to make the
system phase-lock \cite{Acebrn2005,Aeyels2004,Bronski2012,Canale2008,Chopra2009,Ermentrout1985,Ha2010,DeSmet2007,Strogatz2000,Verwoerd2008,Verwoerd2009}.  One
particularly useful result by Dorfler and Bullo \cite{Drfler2011} is an explicit sufficient condition on
the frequency spread that guarantees phase-locking
\begin{equation}
\gamma > \omega_{max} - \omega_{min},
\end{equation}
where $\omega_{max}:=\max\limits_{i}\omega_i$ and $\omega_{min}:=\min\limits_{i}\omega_i$.
Under this condition, the Kuramoto model (\ref{eqn:Kuramoto}) supports
full phase-locking 
for all possible distributions of the natural frequencies supported on
$[\omega_{min},\omega_{max}]$. On the other hand the standard $\ell_1/\ell_\infty$ estimate on the
sum gives a necessary condition on the coupling strength $\gamma$ in
order for the system to support a phase-locked state
\begin{equation}
 \gamma  \geq
 \frac{N}{2(N-1)}\left(\omega_{max}-\omega_{min}\right)\approx \frac{1}{2}\left(\omega_{max}-\omega_{min}\right).
\label{eqn:l1linf}
\end{equation}

From Equation (\ref{eqn:l1linf}) it is easy to see that if $\omega_i$ are 
independent and identically distributed according to a distribution
with unbounded support then in the large $N$ limit one can expect, at
best, partial phase-locking, as the law of large numbers will
guarantee that, with high probability, Equation  (\ref{eqn:l1linf}) will be
violated. To see this note that
\begin{align*}
\mathbb{P}(\max_{i\in\{1\ldots N\}} |\omega_i|<c) = (\mathbb{P}(|\omega_i|<c))^N.  
\end{align*}
If the support of the distribution is unbounded then $
(\mathbb{P}(|\omega_i|<c))<1$ for all $c$ and thus
$\lim_{N\rightarrow \infty}\mathbb{P}(\max_{i \in \{1\ldots N\}} |\omega_i|<c)\rightarrow 0$, so for fixed
coupling strength $\gamma$ full-phase-locking occurs with vanishing
probability in the large $N$ limit. One can, of course, consider
scaling $\gamma$ with $N$ --- this  involves extreme value statistics
of the distribution\cite{Bronski2012} --- but if one is taking $\gamma$ to be
fixed one must consider partial phase-locking or partial entrainment. 

The importance of  partially locked states has been
understood for a long time. The physical arguments of Kuramoto suggest
that the order parameter should
undergo a phase transition at some critical coupling $\gamma^*$, with
amplitude $\propto \sqrt{\gamma-\gamma^*}$: since the amplitude is
small for $\gamma \succsim \gamma^*$ one expects only partial
synchronization. Strogatz gives a nice 
survey in his paper from
2000\cite{Strogatz2000}.  In particular he mentions the Bowen lectures
of Kopell in 1986, where she raises the possibility of doing a
rigorous analysis for large but finite $N$ and then trying to prove a
convergence result as $N\rightarrow\infty.$ The current paper is an
attempt to follows this program. 

The general bifurcation picture described by Kuramoto
has been established  for the continuum model: Strogatz and Mirollo
introduced the continuum model and showed that if the frequencies are
distributed with density $g(\omega)$ then the incoherent state
goes unstable exactly at the critical value $\gamma^*= \frac{2}{\pi g(0)}$  predicted by
Kuramoto \cite{Strogatz1991}. Strogatz, Mirollo and
Matthews \cite{Strogatz1992} showed that below the threshold $\gamma^*$ the evolution decays to an
incoherent state via Landau damping, and Mirollo and Strogatz computed
the spectrum of the partially locked state in the continuum model \cite{DeSmet2007}.
This general picture has been expanded by a number of authors
including Fernandez \cite{Fernandez2015}, Dietert \cite{Dietert2016}
and Chiba \cite{Chiba2018}. See also the review paper of Acebr\'on,
Bonilla, P\'erez Vicente, Ritort and Spigler \cite{Acebrn2005}, particularly section
II.

The partially phase-locked states in the finite N Kuramoto model have received somewhat less attention in the
literature than either fully phase-locked states of the finite N model
 or partially phase-locked
states in the continuum model. Among the finite $N$ results we do
mention the work  of Aeyels and
Rogge \cite{Aeyels2004} and particularly De Smet and
Aeyels \cite{DeSmet2007}. De Smet and Aeyels establish a partial
entrainment result that will be important for the the latter part of this paper. For purposes of this
paper we will draw a distinction between phase-locking and entrainment
(as used by De Smet and Aeyels): we will use partially phase-locked to refer to
a subset of oscillators  which  approximately rotate rigidly. More
precisely a partially phase-locked subset $S$ of oscillators is one for which
\[
 \limsup_{t \rightarrow \infty} |\theta_i(t) - \theta_j(t)| \leq
 \delta(N) \qquad \forall~i,j \in S,
  \]
  where $\delta(N) \rightarrow 0$ as $N \rightarrow \infty$. Typically in this paper $\delta \propto
  N^{-\frac12}$, where $N$ is the total number of
  oscillators. Following De Smet and Aeyels we use partial entrainment to mean that there exists a constant $c$ small but independent of $N$  such that
  \[
 \limsup_{t \rightarrow \infty} |\theta_i(t) - \theta_j(t)| \leq c
 \qquad \forall ~~i,j \in S.
  \]
Obviously this distinction is mainly important in the large $N$
limit. 
  
In this paper we present two independent
but related results. Firstly we consider the question of perturbing
a phase locked solution by adding in additional oscillators that are
not phase-locked to the main group. We define a collection of semi-norms and
associated cylindrical sets in the phase space. We show that under
suitable conditions the semi-norms are decreasing in forward time, and
thus the associated cylindrical sets are invariant in forward
time. The invariance of the cylindrical sets in forward time implies
the existence of a subset of oscillators that remain close in phase
for all time, while the infinite directions of the cylinder
correspond to the degrees of freedom of the remaining oscillators that
are not phase-locked to the group. More precisely, we first consider a Kuramoto model with a small forcing term and prove a standard proposition showing that if the unperturbed Kuramoto problem admits a stable phase-locked solution then the perturbed problem admits a solution that stays near to this phase-locked solution. We then apply
this proposition to the Kuramoto model itself by identifying a subset of oscillators with a small spread in natural frequency and treating the remaining oscillators as a perturbation. This will lead to a sufficient condition for the existence of a partially phase-locked
solution in terms of the infimum over all subsets of oscillators of a
certain function of the frequency spread in that subset. Under such
condition, the number of unbounded oscillators is at most
$N^{1/2}$. Finally we present some supporting numerical
experiments. 

For the second result we reconsider some earlier work of De Smet and
Aeyels \cite{DeSmet2007} in the case where the natural frequencies of
the oscillators are independent and identically distributed random
variables, in the large $N$ limit. We
analyze the condition derived in \cite{DeSmet2007} for the existence
of a positively invariant region and show that in the large $N$ limit
we can find a deterministic condition guaranteeing the existence of a
positively invariant region for sufficiently large coupling constant
$\gamma$.  The theorem shows that, for the coupling strength $\gamma$
sufficiently large and $\omega_i$ chosen independently and identically
distributed from some reasonable distribution then with probability
approaching one as $N \rightarrow \infty$ there exists an entrained
subset of oscillators of positive density. We also get deterministic
upper and lower bounds on the  size of the partially entrained cluster.

\section{Definitions and a partial phase-locking result.}
Our first result is to establish that, given a set of stable phase locked
oscillators, one can add to the system a  second set of
oscillators that do not phase-lock to the first without materially
impeding the phase locking. Before going into details we first give
some intuition why we expect this to be true. The following is
reasonably well-known. Suppose that  an autonomous ODE ${\bm x}_t = {\bm f}({\bm
  x})$ has an asymptotically stable fixed point ${\bm x}_0$   where
the linearization is coercive: ${\bm y}^T \nabla {\bm f}({\bm x}_0)
{\bm y} \leq - c \Vert {\bm y} \Vert^2$. If one makes a sufficiently small
time-dependent perturbation to the ODE, ${\bm x}_t = {\bm f}({\bm
  x}) + \epsilon g({\bm x},t)$, then there will be a small ball
around the former fixed point that is invariant in forward time
(trapping) --
trajectories that begin in the region remain so for all time. To see
this let ${\bm x} = {\bm x}_0 +  {\bm y}$ and note that
\begin{align*}
 & {\bm y}_t = {\bm f}({\bm x}_0 + {\bm y}) + \epsilon {\bm g}({\bm
   x}_0 + {\bm y},t), \\
   & {\bm y}_t \approx \nabla {\bm f}({\bm x}_0) {\bf y} + \epsilon {\bm g}({\bm
     x}_0 + {\bf y},t), \\
  & \frac12 \frac{d}{dt} \Vert {\bm y} \Vert^2 \approx {\bm y}^T \nabla
    {\bm f}({\bm x}_0) {\bm y} + \epsilon {\bm y}^T {\bm g},  \\
   & \frac12 \frac{d}{dt} \Vert {\bm y} \Vert^2 \lesssim - c \Vert {\bm
     y} \Vert^2 + \frac{\epsilon}{2} (\Vert {\bm y}\Vert^2 + \Vert
     {\bm g} \Vert^2).
  \end{align*}
  Thus if $\Vert {\bm y}\Vert $ is  the right size:  large enough that $- (c-\frac{\epsilon}{2}) \Vert {\bm
     y} \Vert^2 + \frac{\epsilon}{2} \Vert {\bm g}\Vert^2 <0 $ but
   small enough to justify $ {\bm f}({\bm x}_0 + {\bm y}) \approx
   \nabla {\bm f}({\bm x}_0){\bm y}$, we find that $\frac{d}{dt}\Vert
   {\bm y}\Vert^2 \leq 0  $ and orbits initially in the ball remain so
   for all time. The intuition, therefore, is that under perturbation
   the fixed point should smear out to an invariant ball of radius
   $\sqrt{\epsilon}$. Similar constructions are used in the PDE
   context to prove the existence of attractors\cite{BG2006,Cees.1993,Otto.2005,Otto.2015,Nicolaenko1985,ralf.2002,Ralf.2014}.
  In the proof of the actual theorem, of course, we
   will take a bit more care but this is the essential idea. 
  
 Our first goal is to define what we mean by a partially phase-locked
solution. To this end we shall define a family of semi-norms $\Vert
\cdot \Vert_S$ indexed by a subset of oscillators $S \subseteq
\{1,2,3,\ldots,N\}$
representing the collection of phase-locked oscillators.

\begin{define}
Given a non-empty index subset $S \subseteq \Omega = \{1,2,3,\ldots,N\}$, we define a semi-norm on a phase vector $\theta$ with respect to $S$ as follows
\begin{equation}
||\theta||^2_S := \frac{1}{|S|}\sum\limits_{i,j\in S, i\le j}(\theta_i-\theta_j)^2,
\label{def:semi_norm}
\end{equation}
where $|S|$ is the cardinality of the set $S$.
\end{define}

\begin{remark}
The open semi-ball $\Vert
\theta \Vert_S< R$ is a cylinder in ${\mathbb R}^N$ that is unbounded in 
$N-|S|+1$ directions and is bounded in the remaining $|S|-1$
directions. The unbounded directions correspond to the $N-|S|$
oscillators that are not phase-locked together with $1$ direction
corresponding to the
common translation
mode $\theta \mapsto \theta + \alpha {\hat{\textbf{1}}}$.

Note that when taking the universal set, i.e., $S = \Omega$, we have 
\begin{equation}
||\theta||_S^2 = \frac{1}{N}\sum\limits_{1\le i\le j\le
  N}(\theta_i-\theta_j)^2 =  \Vert\theta-\langle\theta\rangle \hat{\textbf{1}}\Vert^2.
\end{equation}
In this case, the semi-norm reduces to the usual $\ell_2$ norm modding out by
the translation degree of freedom. 
\end{remark}

Of course these are only semi-norms, not norms, as there is always at
least one null direction. However we will slightly abuse notation by
referring to sets $\Vert \theta \Vert_S < r $ as a ball of radius $r$
since the whole idea is to mod out what is happening in the null
directions. Having defined these semi-norms we can use this to define partial
phase-locking. 

\begin{define}
Let $\mathbb{T}^1 = (-\pi,\pi]$ be a torus and $\mathbb{T}^N$ a
$N$-dimensional torus. Denote 
\begin{itemize}
\item $|\theta_1-\theta_2|$: geodesic distance between $\theta_1 \in \mathbb{T}^1$ and $\theta_2 \in \mathbb{T}^1$.
\item ${\triangle}(\alpha,N) := \{(\theta_1,\theta_2,...,\theta_N) \in \mathbb{T}^N|\max_{i,j=1}^N |\theta_i-\theta_j| < \alpha\}$ for any $\alpha \in [0,\pi]$
\item $\bar{\triangle}(\alpha, N) := \{(\theta_1,\theta_2,...,\theta_N) \in \mathbb{T}^N|\max_{i,j=1}^N |\theta_i-\theta_j| \le \alpha\}$ for any $\alpha \in [0,\pi]$.
\end{itemize}

We say our model (\ref{eqn:Kuramoto}) achieves partial
phase-locking if for some constant vector $\theta^*\in\mathbb{T}^N$, there exists a subset of the oscillators $S$ in (\ref{eqn:Kuramoto}) such as the following is true: the translated phase vector $\tilde{\theta} := \theta - \theta^*$ satisfies $\limsup_{t \rightarrow \infty} |\tilde{\theta}_i(t) - \tilde{\theta}_j(t)| \leq \delta(N)$ for any $i, j \in S$ where $\delta(N) \rightarrow 0$ as $N \rightarrow \infty$. Roughly speaking, if an invariant ball exists for some oscillators while the other oscillators drift away, the dynamical system (\ref{eqn:Kuramoto})
achieves partial phase-locking. In particular, if $S=\{1,2,3,\ldots,N\}$, we say
(\ref{eqn:Kuramoto}) achieves full phase-locking. 
\end{define}

We will need the following result proved by Dorfler and Bullo \cite{Drfler2011}, restated here for convenience:

\begin{thm} [D\"orfler-Bullo] If $\gamma > \gamma_{\text{critical}} := \omega_{max}-\omega_{min}$, then the Kuramoto model (\ref{eqn:Kuramoto}) achieves full phase-locking and all oscillators eventually have a common frequency which is $\omega_{avg} = \frac{1}{N}\sum_{j=1}^N \omega_j$. Also, the set $\bar{\triangle}(\alpha,N)$ is positively invariant for every $\alpha\in[\alpha_{min},\alpha_{max}]$, and each trajectory starting in $\triangle(\alpha_{max},N)$ approaches asymptotically $\bar{\triangle}(\alpha_{min},N)$. Here, $\alpha_{min}$ and $\alpha_{max}$ are two angles which satisfy $\sin(\alpha_{min}) = \sin(\alpha_{max}) = \gamma_{\text{critical}}/\gamma$ and $\alpha_{min} \in [0,\pi/2), \alpha_{max} \in (\pi/2, \pi]$.
\label{thm:DB}
\end{thm}

In order to state our main theorem we first need to define two
functions $g(K,N)$ and $h(K,N)$ that will prove important to the
subsequent analysis: 
 
\begin{define}
 For the Kuramoto model (\ref{eqn:Kuramoto}) with natural frequencies
 $\{\omega_i\}_{i=1}^N$, define two functions:
\begin{align}
&g(K,N) = \min\limits_{S\subset \Omega, |S|=N-K} \max\limits_{i,j\in S} |\omega_i-\omega_j|, \label{def:g} \\
&h(K,N) =  \frac{(N-K)}{N}\sqrt{1-\frac{2K}{(N-K)^{1/2}}}.\label{def:h}
\end{align}
We note that $g(K,N)$ depends implicitly on the set of natural
frequencies $\{\omega_i\}_{i=1}^N,$ and represents the minimum spread
in frequencies over subsets of size $N-K$. The function $h(K,N)$ will
arise in the subsequent analysis and $\gamma h(K,N)$ represents an
estimate of the maximum spread in frequencies for $N-K$ oscillators to be phase-locked. Note that $h(K,N)$ is only defined for $K \leq \frac{\sqrt{16N+1}-1}{8}
\approx 0.5 N^{\frac{1}{2}}$. 
\label{def:g_and_h}
\end{define}

With Definition \ref{def:g_and_h} we are ready to state our main theorem, which gives a sufficient condition on the existence of partially phase-locked states. 

\begin{thm}
\label{thm:Main1}
  Suppose that there exists some integer $K\le \frac{\sqrt{16N+1}-1}{8}$ such
that  $g(K,N) < \gamma h( K, N)\}$, then for some constant vector $\theta^*\in\mathbb{T}^N$, there exists a subset of oscillators $S$ with $|S|=N-K$ such that
\begin{enumerate}
\item {\bf INVARIANCE} There exists a constant $R$ with $R=O(1)$ such that every oscillator with the initial phase condition $||\theta(0)-\theta^*||_S < R$ satisfies $||\theta(t)-\theta^*||_S < R$ for all $t>0$. In other words, the ball $||\theta(t)-\theta^*||_S < R$ is invariant in forward time. 
\item {\bf CONVERGENCE} There exists a constant
  $r=O(\frac{1}{\sqrt{N}})\ll R$ such that
  orbits that begin in the larger ball  $||\theta(0)-\theta^*||_S < R$
  converge to the smaller ball $||\theta(t)-\theta^*||_S < r$ asymptotically.
\end{enumerate} 

\begin{figure} [h]
\centering
      \centering
      \includegraphics[width=.4\linewidth]{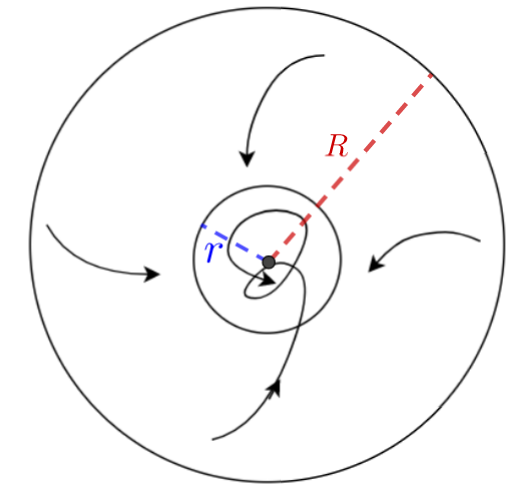}
\caption{Attracting and invariant balls for a subset of $N-K$
  oscillators}
\label{fig:Balls}
\end{figure}

\end{thm}

\begin{remark} 
We make a few remarks about this theorem. Firstly, we can actually derive analytical expressions for the sizes of the invariant and attracting balls, which are $r=\frac{2K\gamma(N-K)^{1/2}}{N|\lambda_2|}$ and $R =
\frac{N|\lambda_2|}{(N-K)\gamma}$. Here, $\lambda_2$ is the second largest eigenvalue of the Jacobian matrix of Equation (\ref{eqn:Kuramoto}) at $\theta^*$. Note that $\lambda_2$ depends implicitly on $\gamma$, and as $\gamma$ increases we expect $\lambda_2$ to become more negative.

Secondly, The integer $K$ represents the number of free or non-phase-locked oscillators. The function $h(K,N)$ is only defined for $K \leq \frac{\sqrt{16N+1}-1}{8} \approx 0.5 N^{\frac{1}{2}}$ for $N$ large, so this theorem can only
guarantee the existence a subset of mutually phase-locked
oscillators with $K \lesssim 0.5N^{\frac{1}{2}}$ oscillators drifting
away. The constant can probably be improved but we think it unlikely
that the scaling 
can be improved without substantially changing the approach. 

Typically we will have $g(K,N)< \gamma h(K,N)$ in an interval, so
there will be a range of integers $K$ for which the inequality is
satisfied. In this situation, we would be primarily interested in the
smallest such $K$ that satisfies the inequality, as this would
represent the largest partially phase-locked cluster. We denote such
$K$ as $K^*$. In other words  $K*$ is the infimum over all  $K$ such
that the inequality $g(K,N)<\gamma h(K,N)$ holds.

When $K=0$, corresponding to no free oscillators, the condition on $\gamma$ in this theorem reduces to  $\gamma_{\text{critical}}(0) = \omega_{max}-\omega_{min}$, which coincides with Theorem \ref{thm:DB} of Dorfler and Bullo\cite{Drfler2011}. Thus this theorem can be viewed as a generalization of their result to the case of partial phase-locking.  
\end{remark}

\section{Proof of Theorem \ref{thm:Main1}}

In this section, we prove our first main result.  A brief sketch of the main idea of the proof is as follows: we first prove a standard proposition: If we take the Kuramoto model in a parameter regime where there is a stable fixed point and
we add a small perturbation, then there is an attracting ball of small radius around the former fixed point. In particular, any initial conditions which begin near the fixed point remain so for all time. We then use this result to study partial phase-locking by considering subsets of oscillators that could potentially phase locked, and considering the remaining oscillators as a perturbation to these candidates for partial phase-locking.

\begin{define}
We say $\theta^*$ is a stable phase-locked solution of
Equation (\ref{eqn:Kuramoto})  with frequencies $\bm{\omega}=(\omega_1, \omega_2, ..., \omega_N)^T$ if
it satisfies 
(\ref{eqn:Kuramoto}) 
\begin{align}
\omega_i = \frac{\gamma}{N} \sum_{j} \sin(\theta_i^* - \theta_j^*)
\end{align}
 and  $J$, the Jacobian matrix at $\theta^*$, i.e,
\[
    J_{ij}(\theta^*) = \begin{cases}
    \frac{\gamma}{N}\cos(\theta_i^*-\theta_j^*), & i \not= j,\\
    -\frac{\gamma}{N}\sum_{k\not=i}\cos(\theta_i^*-\theta_k^*), &  i=j
    \end{cases} 
    \] 
is negative semi-definite with a one dimensional kernel. 
\end{define}

\begin{prop}
  Suppose $\theta^*$ is a stable phase-locked solution. Consider the following perturbed Kuramoto model with perturbation $\epsilon f_i$:
\begin{align}
    \dot{\theta_i} = \omega_i - \frac{\gamma}{N} \sum_{j} \sin(\theta_i - \theta_j) + \epsilon f_i(\theta,t),  \qquad  i = 1,2,...N,
    \label{eqn:pert_Kura}
\end{align}
where $\epsilon$ is a small constant and $f_i's$ are functions bounded by a constant $C$, i.e., $\max\limits_{\theta, t, i} |f_i(\theta, t)|\le C$. Let
\begin{equation}
 \left\{
   \begin{array}{c}
   r(\epsilon) = 2\epsilon CN^{1/2}/|\lambda_2|  \\
   R = |\lambda_2|/\gamma ,\\
   \end{array}
 \right.
\end{equation}
where $\lambda_2<0$ is the second largest eigenvalue of the Jacobian matrix of (\ref{eqn:Kuramoto}) at $\theta^*$. Then for $\epsilon < |\lambda_2|^2/(2CN^{1/2}\gamma)$, the following statements hold:\\
(1) The ball $||\theta(t)-\theta^*||_N < r(\epsilon)$ is invariant in forward time.\\
(2) Every solution with $||\theta(0)-\theta^*||_N < R$ asymptotically converges to the above invariant ball with radius $r(\epsilon)$.
\label{prop:ball}
\end{prop}

\noindent
\begin{proof}
We will make a  standard Lyapunov function calculation: the proof is
sketched here, with details relegated to the Appendix.
We will represent  $\theta$  as  $\theta = \theta^* +
\tilde{\theta}$. Note that by rotational invariance we can assume $\tilde\theta$ is mean zero. Also note that the norm $\Vert \cdot \Vert_\Omega^2$ is equivalent to the standard Euclidean norm $\Vert \cdot \Vert^2$ on the subspace of mean zero functions: if $\tilde \theta$ has mean zero and
$\hat{\boldsymbol{1}}$ is the vector of all ones then $\Vert
\tilde \theta + \alpha \hat{\boldsymbol{1}}\Vert_\Omega^2 = \Vert \tilde \theta
\Vert_\Omega^2 = \sum \tilde \theta_i^2.$ 

First note that we have an upper bound on
$\frac{d}{dt}{||\tilde{\theta}||^2}$ of the following form:
$$\frac{d}{dt}{||\tilde{\theta}||^2} \le 2\lambda_2 ||\tilde{\theta}||^2 + \gamma ||\tilde{\theta}||^3 + 2\epsilon CN^{1/2} ||\tilde{\theta}||.$$
To make $\frac{d}{dt}{||\tilde{\theta}||^2}$ negative, it suffices to require 
\begin{equation}
 \left\{
   \begin{array}{c}
   2\epsilon C N^{1/2} ||\tilde{\theta}|| < |\lambda_2| ||\tilde{\theta}||^2  \\
  \gamma||\tilde{\theta}||^3 < |\lambda_2| ||\tilde{\theta}||^2, \\
   \end{array}
 \right.
\end{equation}

which is equivalent to 
\begin{equation}
\frac{2\epsilon CN^{1/2}}{|\lambda_2|} < ||\tilde{\theta}|| < \frac{|\lambda_2|}{\gamma}.
\end{equation}

Let $r(\epsilon) = 2\epsilon CN^{1/2}/|\lambda_2|$ and $R =|\lambda_2|/\gamma$, then by Gronwall's inequality\cite{Khalil1993}, the semi-norm of $\tilde{\theta}$ is
exponentially decreasing when $\tilde{\theta}$ is in the annulus of
radii $r(\epsilon)$ and $R$, and then stays in the ball of radius
$r(\epsilon)$ forever. So statements (1) and (2) are proved.
\end{proof}

Now, we use Proposition \ref{prop:ball} to prove Theorem \ref{thm:Main1}.

\begin{proof}
For any integer $0\leq K < N$, consider
$N-K$ oscillators in the Kuramoto model (\ref{eqn:Kuramoto}). By
changing the order of labels, we can, without loss of generality,
focus on the first $N-K$ oscillators and study the conditions under
which they will stably phase-lock.  The evolution
can be written as follows
\begin{align}
\dot \theta_i &= \omega_i - \frac{\gamma}{N} \sum_{i=1}^N \sin(\theta_i - \theta_j) \\
  &= \omega_i - \frac{\gamma}{N} \sum_{j=1}^{N-K} \sin(\theta_i - \theta_j) - \frac{\gamma}{N} \sum_{j=N-K+1}^{N} \sin(\theta_i -
    \theta_j) \\
  &= \omega_i - \frac{\tilde\gamma}{N-K} \sum_{j=1}^{N-K} \sin(\theta_i - \theta_j) + \epsilon f_i
\end{align}
where $\tilde\gamma= \gamma\frac{N-K}{N}$ is a modified coupling strength on the first $N-K$ oscillators and $\epsilon f_i$ represents the effect of the remaining $K$
oscillators. Then we have $\epsilon = \frac{\gamma}{N}$, $f_i = \sum_{j=N-K+1}^{N} \sin(\theta_j - \theta_i) \le K$. The strategy is to treat the effect of the remaining
$K$ oscillators as a perturbation and then apply Proposition (\ref{prop:ball}). 

We first consider the unperturbed problem
\begin{align} 
    \dot{\theta_i} = \omega_i - \frac{\tilde{\gamma}}{N-K} \sum_{j=1}^{N-K} \sin(\theta_i - \theta_j), \qquad   i = 1,2,...,N-K. \label{eqn:unpert}
\end{align} 
Define
\[ \gamma_0 = \max_{i,j=1}^{N-K} |\omega_i-\omega_j|.
 \] By Theorem
\ref{thm:DB} if the spread in frequencies satisfies 
\begin{equation}
 \gamma_0 < \tilde{\gamma} = \frac{\gamma(N-K)}{N},
\end{equation}
then Equation (\ref{eqn:unpert}) phase-locks; the
set $\bar{\triangle}(\alpha)$ is positively invariant for every
$\alpha \in [\alpha_{min},\alpha_{max}]$, and each trajectory starting
in $\triangle(\alpha_{max})$ approaches asymptotically
$\bar{\triangle}(\alpha_{min})$, where $\alpha_{min}\in [0,\pi/2)$,
$\alpha_{max}\in (\pi/2, \pi]$ and $\sin(\alpha_{min}) =
\sin(\alpha_{max}) = \frac{\gamma_0}{\tilde{\gamma}}$. From these, it
is clear to see that under a rotating frame with frequency
$\omega_{avg}$, Equation (\ref{eqn:unpert})
has a fixed point $\theta^*$ such that $\theta^* \in \bar{\triangle}(\alpha_{min})$.

Suppose $L$ is the Jacobian matrix of (\ref{eqn:unpert}) at the fixed point $\theta^*$, i.e,
\[
    L_{ij} = \begin{cases}
    \frac{\tilde \gamma}{N-K} \cos(\theta_i^* -
                   \theta_j^*), & i \not= j,\\
     - \frac{\tilde \gamma}{N-K}\sum_k
                   \cos(\theta_i^* - \theta_k^*),
                   & i = j.
    \end{cases} 
    \] 

Since $\theta^* \in \bar{\triangle}(\alpha_{min})$ and $\alpha_{min}
\in [0,\frac{\pi}{2})$, we have $\cos(\theta^*_i-\theta^*_j) > 0$ and
$L$ is a negative semidefinite Laplacian matrix with eigenvalues
$\lambda_1 = 0 > \lambda_2 \ge \lambda_3 \ge ... \ge
\lambda_{N-K}$, so the solution is stably phase-locked.  

We next consider the effects of the perturbation terms $\epsilon f_i$ where $\epsilon = \frac{\gamma}{N}$, $f_i =
\sum_{j=N-K+1}^{N} \sin(\theta_j - \theta_i) \le K$. Proposition
\ref{prop:ball} guarantees the existence of an invariant ball for the
first $N-K$  oscillators when 
\begin{equation}
\epsilon = \frac{\gamma}{N} <
       \frac{|\lambda_2|^2}{2K(N-K)^{1/2}\tilde{\gamma}},
\end{equation}
or equivalently, when
\begin{equation}
    \gamma < \sqrt{\frac{1}{2K(N-K)^{3/2}}}\cdot N|\lambda_2|.
    \label{eqn:Condition}
\end{equation}

The eigenvalue $\lambda_2$ depends implicitly on $\gamma$ so we
need a lower bound on the magnitude of $\lambda_2$ in order to close
the argument and guarantee that (\ref{eqn:Condition}) can be satisfied.
Since the kernel of $L$ is spanned by $(1,1,1,\ldots,1)$
we can consider the operator $-L$ acting on the space of mean-zero
vectors. For any $x$ with $\sum\limits_i x_i = 0$, we have, on the one hand,
\begin{align*}
  x^{T}(-L)x &= \frac{\gamma}{N}\sum\limits_{i,j} \cos(\theta_i^* - \theta_j^*) x_{i}^2 - \frac{\gamma}{N}\sum\limits_{i,j}
 \cos(\theta_i^* - \theta_j^*) x_{i} x_{j} \\
    &= \frac{\gamma}{2N}\sum\limits_{i,j} \cos(\theta_i^* - \theta_j^*)(x_i-x_j)^2 \\
    & \geq  \frac{\gamma}{2N}\min_{i,j} \cos(\theta_i^* - \theta_j^*)\sum\limits_{i,j} (x_i-x_j)^2\\
    & = \frac{\gamma}{N} \min_{i,j} \cos(\theta_i^* - \theta_j^*)\left((N-K)\sum\limits_{i}x_i^2-\sum\limits_{i,j}x_ix_j\right)\\
    & = \frac{\gamma(N-K)}{N} \min_{i,j} \cos(\theta_i^* - \theta_j^*)\Vert x\Vert^2\\
  & \geq  \frac{\gamma(N-K)}{N}  \sqrt{1-\frac{\gamma_0^2}{\tilde{\gamma}^2}} \Vert x\Vert^2.
  \end{align*}
  On the other hand, 
  \begin{align*}
      x^{T}(-L)x \le \frac{\gamma}{2N}\sum_{i,j}(x_i-x_j)^2 = \frac{\gamma(N-K)}{N}\Vert x\Vert^2.
  \end{align*}
 Therefore we have the inequality
\begin{align}
\frac{\gamma(N-K)}{N} \sqrt{1-\frac{\gamma_0^2}{\tilde{\gamma}^2}} \le |\lambda_2| \le \frac{\gamma(N-K)}{N}.\label{eqn:LambdaIEQ} 
\end{align}

Combining Equations (\ref{eqn:Condition}) and (\ref{eqn:LambdaIEQ}), we can conclude that an invariant ball for the first $N-K$ oscillators with radius $R=\frac{|\lambda_2|N}{\gamma(N-K)}$ exists when 
\begin{equation}
\gamma_0 = \max_{i,j=1}^{N-K} |\omega_i-\omega_j| < \tilde{\gamma}\sqrt{1-\frac{2K}{(N-K)^{1/2}}}.
\label{ineqn: g<h}
\end{equation}
Therefore, we have proven the first part of Theorem \ref{thm:Main1}. In fact, since the above argument holds regardless of the subset of oscillators we choose, we can go through every subset holding $N-K$ elements and target the one with the smallest $K$ such that (\ref{ineqn: g<h}) holds. So we derive a sufficient condition: $g(K,N)\le \gamma h(K,N)$, where functions $g$ and $h$ are respectively defined in (\ref{def:g}) and (\ref{def:h}). The existence of an invariant ball of $N-K^*$ oscillators where 
\begin{equation}
    K^* := \min\limits_K\{K\in \mathbb{N}:g(K,N) \le \gamma h(K,N)\}
    \label{eqn:Kstar}
\end{equation}
is guaranteed.

Similarly as Proposition \ref{prop:ball}, it can be concluded that if $||\theta(0)-\theta^*||_S < R$, then all the oscillators in $S$ asymptotically converges to the invariant ball $||\theta(t)-\theta^*||_S < r$, where $r=\frac{2
  \gamma K(N-K)^{1/2}}{N|\lambda_2|}$. Therefore, we have a proof for the second part of Theorem \ref{thm:Main1}.
\end{proof}

\section{Numerical Examples}

In this section we present several numerical experiments on the Kuramoto
model (\ref{eqn:Kuramoto}) to illustrate our first theorem.
In the first two experiments all of the oscillator frequencies are
chosen to be i.i.d. Gaussian random variables with small variance
except for one or two whose natural frequency is chosen to be
large compared with
the other oscillators.
In the last experiment we consider a case where all oscillators have independent Cauchy distributed natural frequencies.

\begin{example}[One free oscillator]

The first experiment depicts a case with $N=20$ oscillators with coupling strength $\gamma=1$. The frequencies $\omega_1, \omega_2, ..., \omega_{19}$ are chosen to be normal random variables with mean 0 and variance
$\frac{\gamma}{N}$, and the frequency $\omega_{20}$ is chosen to be $\gamma+0.1$. One can easily check from the definition that $K^*=1$, meaning there exists at most one free oscillator. The cluster of nineteen phase-locked 
oscillators eventually moves at a common angular frequency $\bar{\omega}$. We use the change of variables $\tilde{\theta}_i(t) = \theta_i(t)-\bar{\omega}t$ for $i=1,2,...,N$ to work in a frame of reference corotating with the phase-locked cluster. 
With a slight abuse of notation, we rewrite $\tilde{\theta}$ as $\theta$. The left graph in Figure \ref{fig:free_1_N_20_12} exhibits the evolution of the phases $\theta_i$'s on the real line with respect to time $t$ under the rotation frame; the right graph 
represents the phase trajectories on the torus. It can be seen that, as expected, there exists a phase-locked cluster of 19 oscillators depicted by the blue curves, and a single free oscillator whose trajectory is depicted by the red curve.

\begin{figure} [H]
        \includegraphics[width=2.3in, height=1.8in]{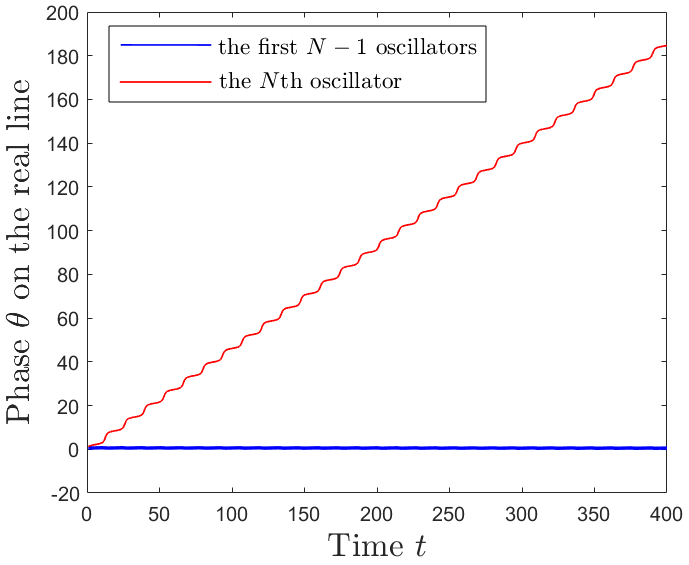} 
        \includegraphics[width=2.3in, height=1.8in]{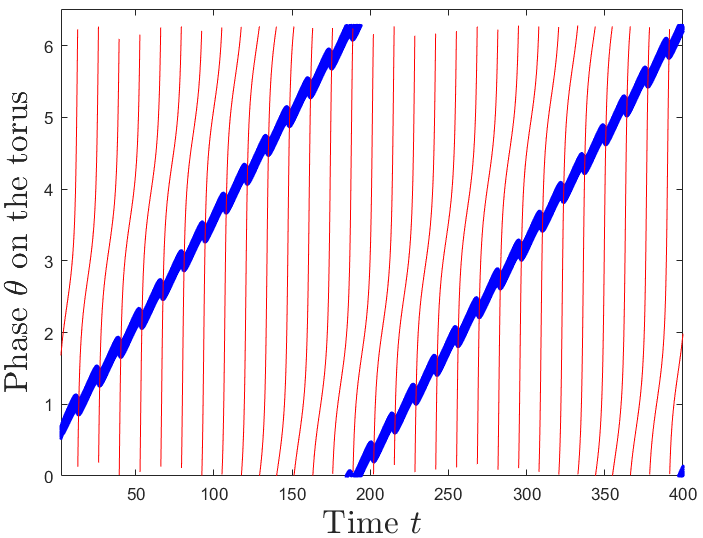}
 \caption{A cluster of 19 phase-locked oscillators and 1 free oscillator.}
\label{fig:free_1_N_20_12}
\end{figure}

Figure \ref{fig:free_1_N_20_3} represents the same experiment from Figure \ref{fig:free_1_N_20_12}, but we have moved to a frame that is co-rotating with the phase-locked cluster and rescaled the graph to more clearly represent the dynamics of the cluster. 
 One can clearly see that after an initial transient the phase-locked cluster settles down to something that appears to be periodic. It is clear that there is a periodic disturbance of the cluster when the free oscillator passes through, although this is not sufficient to break up the cluster. To make this a little more precise, we first computed the frequency of the free oscillator via $\omega_{\text{eff}} = \frac{\theta(T)-\theta(T/2)}{T/2}$ where $T=1000$ is the total running time. This calculation gave $\omega_{\text{eff}} = 0.4483$. Next we took the Fourier transform of the trajectory of one of the oscillators in the locked cluster, excluding the initial transient region. The results are depicted in  Figure \ref{fig:free_1_N_20_3}, which shows the one-sided spectral power density for a single trajectory. One can see that the trajectories are effectively periodic -- the  spectrum has peaks at integer multiples of fundamental frequency $\xi \approx .0714$, and that $\omega_{\text{eff}}\approx 2\pi \xi$, as expected. 

\begin{figure}[H]
    \centering
    \includegraphics[width=2.8in, height=2.2in]{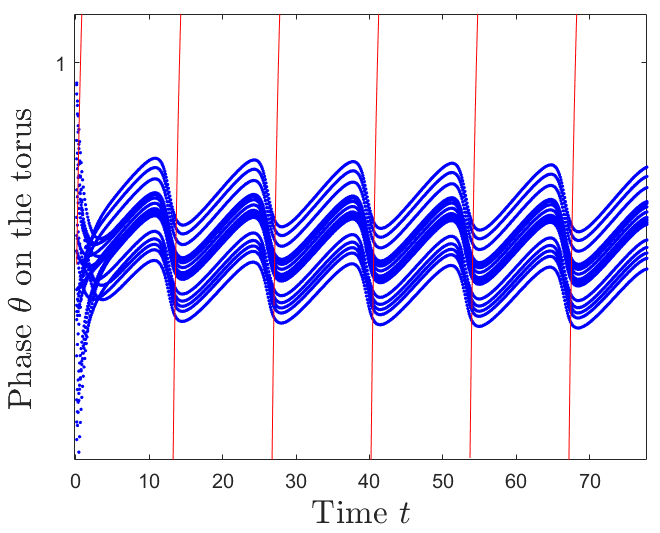}
\caption{Phase trajectories on the torus under a rotated frame}
\label{fig:free_1_N_20_3}
\end{figure}

\begin{figure} [H]
    \centering
    \includegraphics[width=2.8in, height=2.2in]{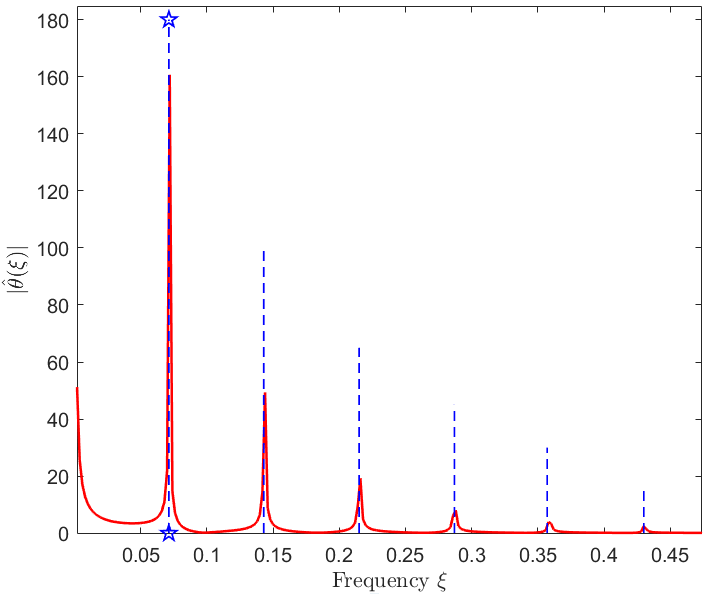} 
     \caption{Single-sided amplitude spectrum of a phase-locked  trajectory}
    \label{fig:free_1_N_20_4}
\end{figure}

\end{example}

\begin{example}[Two free oscillators]

In this example we still consider a system of $N=20$ oscillators with coupling strength $\gamma=1$, but instead choose the frequencies of two of the oscillators to guarantee that they do not phase-lock to the rest. More precisely the frequencies $\omega_1, \omega_2, ..., \omega_{18}$ are
chosen to be Gaussian random variables with mean 0 and variance $\frac{\gamma}{N}$, and the two free oscillators are chosen to have
frequencies $\omega_{19}=\gamma + 0.1$ and $\omega_{20} = 1.5\gamma + 0.01$.
As expected $K^* = 2$, and as before we work in the coordinate system
that rotates with the mean frequency of the cluster of $18$
oscillators. The results of a numerical simulation are depicted in
Figure \ref{fig:free_2_N_20_12}. As in the first experiment we see a stable cluster of
eighteen oscillators  with quasi-periodic disturbances as the two free
oscillators pass through the cluster. 

\begin{figure} [H]
 \includegraphics[width=2.3in, height=1.8in]{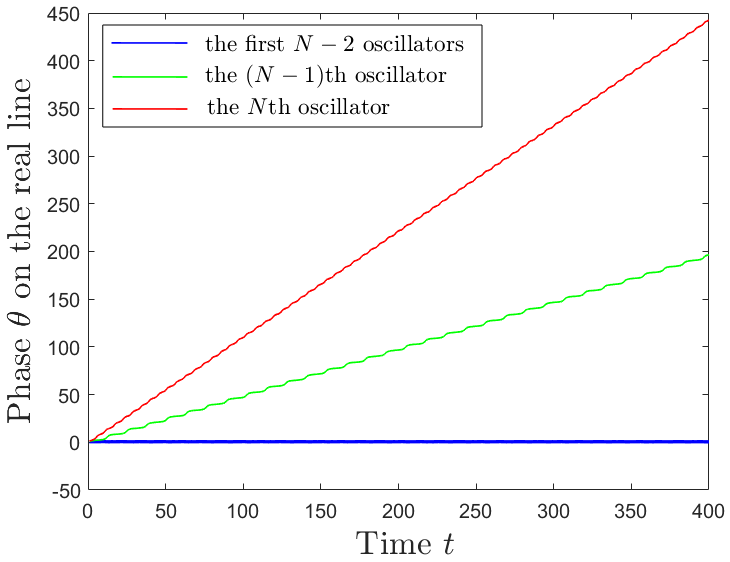}
  \includegraphics[width=2.3in, height=1.8in]{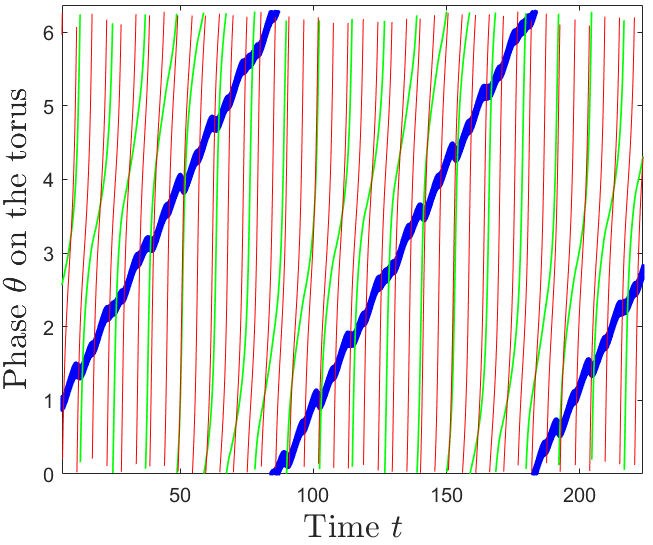}
\caption{A cluster of 18 phase-locked oscillators and 2 free oscillators.}
\label{fig:free_2_N_20_12}
\end{figure}

The left graph in Figure \ref{fig:free_2_N_20_12} shows the phases of the  oscillators in the cluster, which appear to be quasi-periodic. The effect of these two oscillators on the phase-locked ones can be seen from the right graph in Figure \ref{fig:free_2_N_20_12}. Enlarging a portion of this graph and redrawing the trajectories in the co-rotating frame gives Figure \ref{fig:free_2_N_20_3}.

\begin{figure} [H]
    \centering
    \includegraphics[width=2.8in, height=2.2in]{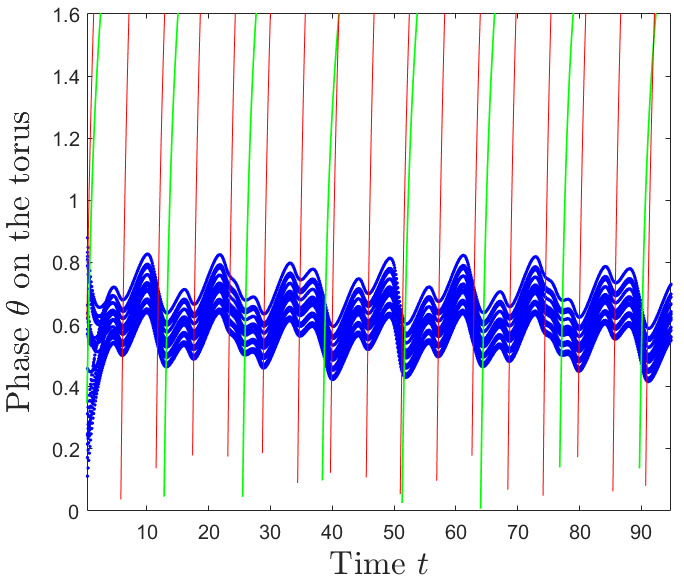}
    \caption{Phase trajectories on the torus under a rotated frame}
\label{fig:free_2_N_20_3}
\end{figure}

In a similar manner to the first experiment we expect a relation between the fundamental frequencies of the phase trajectory of a phase-locked oscillator $\xi$ and the angular frequencies of the free oscillator $\tilde{\omega}$: $\tilde{\omega} = 2\pi\xi$, Once again we compute 
the Fourier transform of one of the trajectories in the phase-locked cluster  and obtain Figure \ref{fig:free_2_N_20_4}.
\begin{figure} [H]
    \centering
    \includegraphics[width=2.8in, height=2.2in]{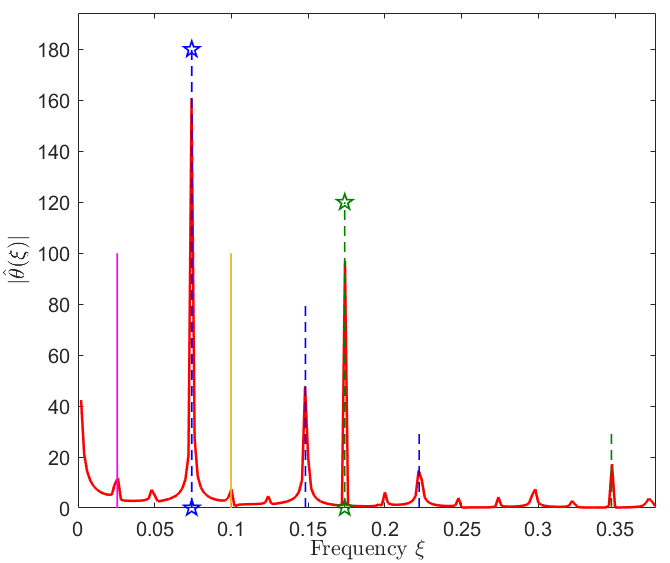} 
     \caption{Single-sided amplitude spectrum of a locked phase trajectory}
    \label{fig:free_2_N_20_4}
\end{figure}

As in the previous experiment we also computed the effective frequencies by $\omega_{\text{eff};19}= \frac{\theta_{19}(T)-\theta_{19}(T/2)}{T/2}$, and found $\omega_{\text{eff};19}\approx 0.4656$ and $\omega_{\text{eff};20}\approx 1.0920$. This agrees well with what we found by computing the Fourier transform of one of the trajectories of an oscillator in the phase-locked cluster. 
The fundamental frequencies, as seen in Figure \ref{fig:free_2_N_20_4}, are $\xi_{19}= 0.0741\approx \omega_{\text{eff};19}/2\pi$ and $\xi_{20} = 0.1738\approx \omega_{\text{eff};20}/2\pi$, associated with the two highest peaks denoted by the dashed lines with the star markers, in agreement with the direct numerical measurement. Since the two free oscillators have incommensurate frequencies we would expect to see many smaller peaks associated with various linear combinations of the fundamental frequencies. We have marked the integer multiples of $\xi_{19}$ with blue lines, and multiples of $\xi_{20}$ with green, as well as a couple of other peaks corresponding to other linear combinations. 
 For instance, the pink line denotes  a frequency of $ \xi_{20} - 2 \xi_{19}$ and the yellow line a frequency $\xi_{20} - \xi_{19}$.  We will not label all of the frequency peaks but all of them correspond to small integer combinations $j \xi_{19} + k \xi_{20}$ with $|j|,|k| \leq 3$. 
 
 \end{example}

\begin{remark}
In the previous two experiments, with one and two free oscillators,
the solutions appeared to be periodic and quasi-periodic
respectively. It is worth noting that it would probably be quite
difficult to prove the existence of a periodic or quasi-periodic
solution. Even if one were able to do so a linear stability analysis
of the solution would likely be highly non-trivial. In the case of a
periodic solution the stability analysis would involve a Floquet problem;
these types of problems are difficult to solve in any but the simplest
of cases. The spectrum of quasi-periodic operators is even more
difficult to understand: in the case of a quasi-periodic Schr\"odinger
operator the spectrum typically lies on a Cantor set\cite{JS1994}, rather than 
simple bands and gaps as in the periodic case. However by showing the
existence of a small exponentially attracting ball we can answer the
same physical question in a much easier way.   
\end{remark}

\begin{example}[Cauchy distributed oscillators]
The first two numerical experiments were instructive but obviously somewhat
contrived in that we picked one or two of the oscillators frequencies
by hand to ensure that we had some free oscillators. 

In this experiment we take $N=500$
oscillators with coupling strength $\gamma=5$. The frequencies
$\omega_1, \omega_2, ..., \omega_{500}$ were chosen to be standard Cauchy
random variables with  constant scale $0.01$, i.e.,  $\omega_i \sim
0.01\cdot Cauchy(0,1)$.
Of course Cauchy random variables have very broad tails, so we expect
large outliers to be relatively common (as compared with, say, a
Gaussian distribution). In the experiment depicted here,
$\omega_{\max}-\omega_{\min} = 7.2161 > \gamma=5$, so the necessary
condition for full phase-locking is not satisfied.  However, partial
phase-locking is guaranteed
if there exists some integer $K$ such that $g(K)<h(K)$. 

\begin{figure}[H]
    \centering
    \includegraphics[width=2.2in, height=1.8in]{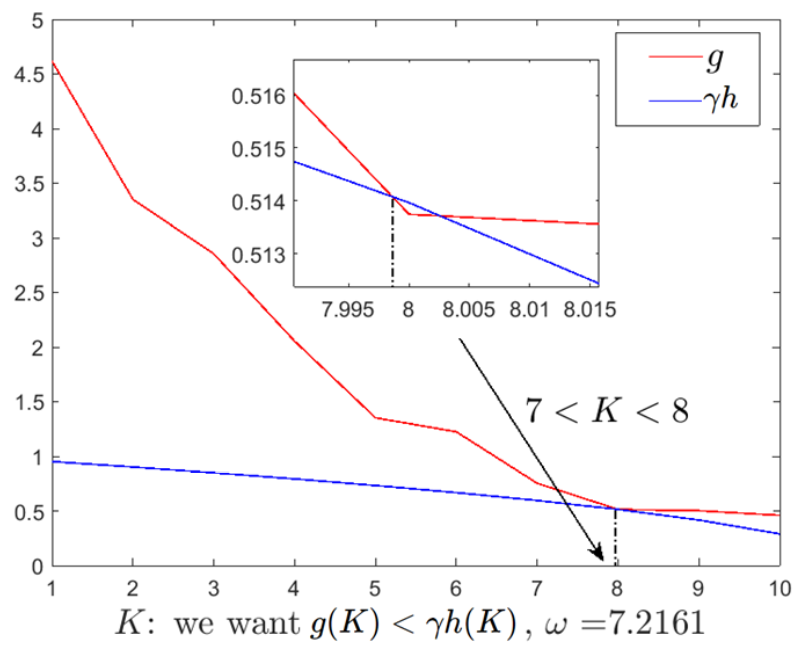}
    \includegraphics[width=2.4in,height=1.8in]{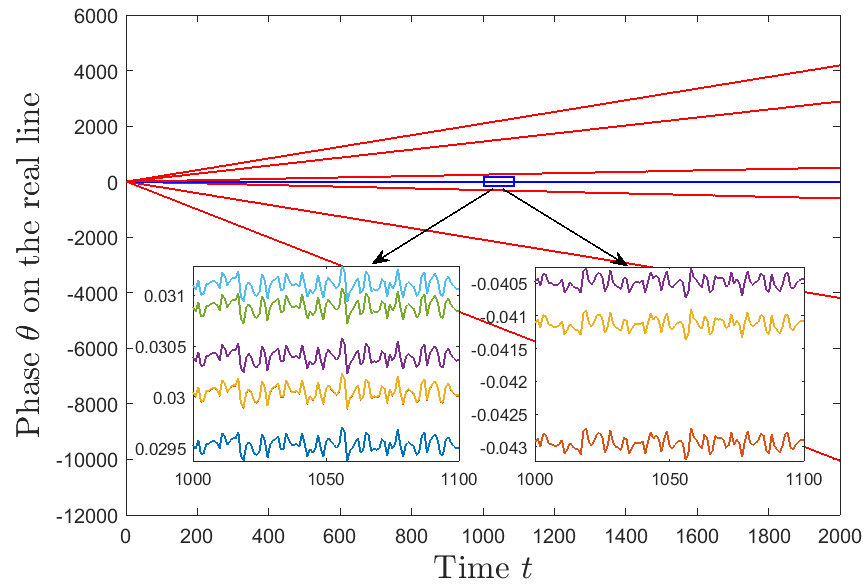}
\caption{Partially phase-locked oscillators with Cauchy distributed frequencies.}
  \label{fig:Exp5}
\end{figure}

The left graph in Figure \ref{fig:Exp5} shows the graphs of
functions $g(K,N)$ (computed directly from the random frequency vector
$(\omega_1,\omega_2,\ldots,\omega_{500})$) and $\gamma h(K,N)$ with respect to $K$. There is a very small region in which the inequality 
$g(K)<\gamma h(K)$ holds, from about $7.998$ to about $8.003$. This guarantees
the existence of  a phase-locked cluster of at least $N-8=492$ oscillators.
The theorem does not really say much about the basin of attraction,
except to guarantee that it has radius at least $O(1).$ 
The right graph in Figure \ref{fig:Exp5} shows the evolution
of the oscillator phases $\theta_i$ with respect to time $t$. In practice we see that the size of the 
phase-locked cluster is somewhat larger than the minimum guaranteed by the theorem: there are actually 494 phase-locked oscillators and 6 free oscillators. 
The red curves represent the trajectories of 494 phase-locked oscillators while the blue curves represent 6 free oscillators. 
\end{example}

\section{Almost sure Entrainment}

Our goal in this section is to understand the probability of partial
entrainment in the Kuramoto model with randomly distributed
frequencies, particularly in the large $N$ limit. The results in the previous section used a relatively strong
definition of partial phase-locking, in that we required a subset of
oscillators to remain close to an equilibrium configuration. This
resulted in fairly strong control on $\Vert \theta- \theta^*\Vert_S$;
however while it allowed a large number of non-phase-locked
oscillators  the percentage as a fraction of the total number had to
remain small. In considering the limit $N \rightarrow \infty$
one would really like to allow the possibility that a fixed
percentage of the oscillators, possibly small but independent of $N$,
would fail to phase-lock.  To this end we utilize a very pretty result of De Smet
and Aeyels \cite{DeSmet2007} that guarantees that a subset of
oscillators remains close to one another, while not necessarily being
close to any fixed configuration: partial entrainment.  

\begin{thm}[Aeyels-DeSmet]
For the finite $N$ Kuramoto model (\ref{eqn:Kuramoto}), if \begin{equation}
\min\limits_{S\subset \Omega, |S|=N-K} \max\limits_{i,j\in S}
|\omega_i-\omega_j| <
\gamma\sqrt{\frac{N}{N-K}}\left(\frac{2N-4K}{3N}\right)^{\frac{3}{2}},
\label{ineqn:Aeyels}
\end{equation}
then there exists a subset $S\subset \{1,...,N\}$ with $|S|=N-K$ such
that
there is an invariant region: 
$$\exists C_S>0 \ s.t. \ |\theta_i(t)-\theta_j(t)| < C_S, \ \forall t\ge 0, \ \forall i,j\in S,$$
where $C_s=2\arcsin{\sqrt{\frac{N-2K}{6(N-K)}}}$, i.e., the Equation (\ref{eqn:Kuramoto}) achieves partial entrainment for at least $N-K$ oscillators.
\label{thm:Aeyels}
\end{thm}

\begin{remark}
The above result is very strong, in the
sense that it can in principle establish entrainment when a positive fraction  (up to roughly  $1/2$) of the
oscillators are free. This is what one would expect from experiments,
applications, and the original physical arguments of Kuramoto. On the
other hand it does not give very much information about the dynamics.
While the angles of the entrained subset of oscillators are guaranteed
to remain close to one another there can in principle be $O(1)$
changes in the relative positions of the oscillators, and thus 
the order parameter is not guaranteed to be constant. One expects
that, on average, the free oscillators will not contribute to the
order parameter (though there is not proof of that) but even defining
a ``reduced'' order parameter based only on the entrained oscillators
the most that one can say is that the order parameter is bounded from
above and below. We will discuss this further in the conclusions section.
\end{remark}

If we denote the right-hand side of inequality (\ref{ineqn:Aeyels}) as $\tilde{h}(K, N)$ and let $\rho = \frac{K}{N}$ represent the density of unlocked oscillators, then it is clear that in this new variable $\rho$ that
\begin{align}
    & g(\rho)=\min\limits_{S\subset \Omega, \frac{|S|}{N}=1-\rho}
    \max\limits_{i,j\in S} |\omega_i-\omega_j|,\\
    & \tilde{h}(\rho) =
    \sqrt{\frac{1}{1-\rho}}
    \left(\frac{2-4\rho}{3}\right)^{\frac{3}{2}}.
\end{align}
In terms of $\rho$, the
inequality (\ref{ineqn:Aeyels}) becomes $g(\rho) < \gamma \tilde{h}(\rho)$. Note
that the function $\tilde{h}(\rho)$ is only well-defined  when $\rho \le \rho_{max}=\frac{1}{2}$. Now we are ready to state our second main result as follows.

\begin{thm}
  \label{thm:Main2}
Consider the Kuramoto model (\ref{eqn:Kuramoto}) where  the natural frequencies $\{\omega_i\}_{i=1}^N$ are chosen
independently and identically distributed from a distribution with 
the following properties
\begin{itemize}
\item The distribution has a density $f(\varpi)$ that is symmetric and
  unimodal with support on the whole line --
  the density is increasing on ${\mathbb R}^-$ and decreasing on
  ${\mathbb R}^+.$
  \item The maximum of the density occurs at $\varpi=0$.
 \end{itemize}

Define the function $g_{\infty}(\rho)$ implicitly by 
\begin{equation}
\int_{-\frac{g_{\infty}}{2}}^{\frac{g_{\infty}}{2}} f(\varpi) d\varpi = 1 - \rho,
\label{eqn:ginfty}
\end{equation}
and the function $\tilde h(\rho)$ by
\[
  \tilde{h}(\rho) = \sqrt{\frac{1}{1-\rho}}
  \left(\frac{2-4\rho}{3}\right)^{\frac{3}{2}}.
  \]
Let $\gamma^*$ be the smallest value of $\gamma$ such that there
exists a solution to 
\begin{equation}
   g_{\infty}(\rho) = \gamma \tilde{h}(\rho) \qquad \rho \in (0,\frac{1}{2}].
   \label{eqn:gh}
\end{equation}
Then $\gamma^*$ is a threshold coupling strength for partial
entrainment in the following sense: let ${\mathbb P}_{N,\gamma}$ denote the probability that the Kuramoto
model admits a partially entrained state with $O(N)$ oscillators. 
Then 
\[
\lim_{N \rightarrow \infty} {\mathbb P}_{N,\gamma} = 1 \qquad \forall \gamma > \gamma^*.
\]

Moreover we have bounds on the size of the largest partially entrained
cluster: if
 $N_{\text{cluster}}$
denotes the number of the oscillators belonging
to the largest partially entrained cluster then
\[
 1-\rho_{\text{min}} \leq \frac{N_{\text{cluster}}}{N} \leq
  \int_{-\gamma}^\gamma f(\varpi) d\varpi. 
  \]
  Here, $\rho_{\text{min}}$ is defined as the smallest $\rho$-coordinate of the
intersection points of $g_{\infty}(\rho) $ and $\gamma\tilde{h}(\rho)$.
The inequality holds in the sense that
\begin{align}
    & \lim_{N \rightarrow \infty}{\mathbb P}(N_{\text{cluster}}  \ge (1-\rho_{\text{min}}) N -O(N^{\frac 12 + \epsilon})) = 1, 
    \label{eqn:ch2_partial_entrain_thm_1} \\
    & \lim_{N \rightarrow \infty}{\mathbb P}(N_{\text{cluster}}  \le N \int_{-\gamma}^\gamma f(\varpi) d\varpi +O(N^{\frac 12 + \epsilon})) = 1.
    \label{eqn:ch2_partial_entrain_thm_2}
\end{align}
\end{thm}

\begin{remark}
This is, of course, a sufficient condition ($\gamma>\gamma^*$) for partial phase-locking and not a necessary one. Of course based on what is known about the continuous Kuramoto model and the physical arguments on the finite $N$
Kuramoto model one expects (and the numerics to be presented later
support this) that partial entrainment occurs for much smaller values of
$\gamma$ than are required by the theorem. 

As far as the hypotheses go, the second condition that the maximum
of the density of the distribution occurs at $\omega=0$  can be assumed w.l.o.g. by working in a co-rotating
frame. In the first  condition the assumption of symmetry is not
really required, and was adopted mostly for ease of exposition, but
the assumption that the density is monomodal enters into the proof in
a more substantial way. This will be discussed later.  

By the definition of $g_{\infty}$, it is clear that $g_{\infty} = 2F^{-1}(1-\frac{\rho}{2})$. Under the assumptions of symmetry and unimodality it is easy to compute that $g_{\infty}(\rho)$ is a decreasing functions with a positive second derivative. It is also easy to compute that $\tilde{h}(\rho)$
is a decreasing function with a positive second derivative. In fact, if one can show $(g_{\infty} - \tilde{h})(\rho)$ is a convex function when $\rho\le \frac{1}{2}$, then it follows that these functions can be equal, $g_{\infty}(\rho)=\tilde{h}(\rho)$, at at most two distinct values of $\rho$, implying that in the continuum limit the
range of possible entrained cluster sizes is an interval. Plus, as the coupling strength $\gamma$ increases, $\rho_{\text{min}}$ decreases until the first intersection point vanishes, which implies that partial synchronization becomes full synchronization. For instance, when $\omega_i$'s follow standard Gaussian distribution, the graph of the functions $g_{\infty}$ and $\gamma\tilde{h}$ is shown below:

\begin{figure}[H]
    \centering 
    \includegraphics[width=2.8in, height=2.2in]{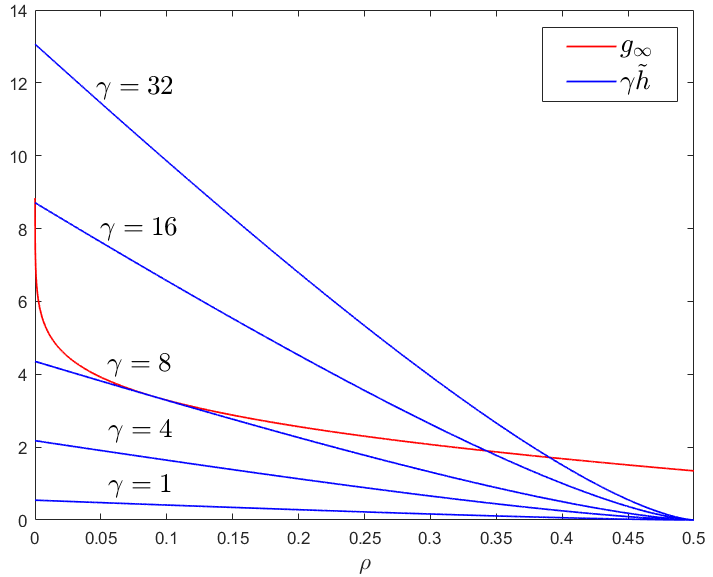}
    \caption{Intersections of $g_{\infty}$ and $\gamma\tilde{h}$ for Gaussian distribution.}
    \label{fig:gh}
\end{figure}
\end{remark}

To prove Theorem \ref{thm:Main2}, we first prove that under the assumptions on the distribution of the  $\omega_i$, in the limit $N \rightarrow \infty$ the function $g(\rho)$ tends to a deterministic function $g_{\infty}(\rho)$, which is Proposition \ref{prop:glimit} stated below. Then with Proposition \ref{prop:glimit} and Theorem \ref{thm:Aeyels}, it is straightforward to derive Theorem \ref{thm:Main2}.

\begin{prop}
Suppose that the natural frequencies $\{\omega_i\}_{i=1}^N$ are independent, identically distributed random variables satisfying  the assumptions in Theorem \ref{thm:Main2}, with $f(\varpi)$ the probability density function and $F(\varpi)$ the cumulative distribution function. Then with high probability $g(K,N)$ converges to a deterministic
function $g_{\infty}(\rho)$ defined by the Equation \ref{eqn:ginfty}. More precisely, we have the estimate 
\begin{equation}
\lim_{N \rightarrow \infty}{\mathbb P} (|g(K,N) - g_{\infty}(\frac{K}{N})| \leq N^{-\frac12+\epsilon}) = 1.
\label{eqn:glimit}
\end{equation}
\label{prop:glimit}
\end{prop}

\begin{proof}[Sketch of proof]
Define $a = F^{-1}(1-\frac{\rho}{2})$ so that we have $g_\infty = 2a$ and define $\delta=\frac{1}{2}N^{-\frac{1}{2}+\epsilon}$. First, using the law of large number theorem, one can easily show $g(\rho)\le 2(a+\delta)$ with probability one. What is less obvious to show is that $g(\rho)\ge 2(a-\delta)$ with probability one where $\delta = \frac{1}{2}N^{-\frac{1}{2}+\epsilon}$ and $\epsilon>0$. In other words, we need to prove 
\begin{equation}
 \mathbb{P}(A) \to 0\ \text{as}\ N\to\infty,  
\end{equation}
where $A$ is the event that ``there exists an interval with length $L=2(a-\delta)$ containing more than $(1-\rho)N$ points". Notice that if no intervals of Length $L$ with $\omega_k$ at an endpoint contain more than $m$ points then no any other interval does. So we can only focus on $N$ intervals $\{[\omega_i, \omega_i+L]: i=1,2,...,N\}$. Moreover, the interval centered at zero maximizes the probability that a point lies in the interval, i.e., $I=[-L/2, L/2]$ gives the largest $\mathbb{P}(x\in I)$ among all intervals of length $L$. Based on these observations, it is not hard to see 

\begin{equation}
    \mathbb{P}(A) \le N\sum\limits_{M=\ceil{(1-\rho)N}}^N \binom{N}{M} p^M(1-p)^{N-M},
    \label{ineqn:ch2.2}
\end{equation}
where $1-\rho = \int_{-a}^a f(x)dx$ and $p=\int_{-L/2}^{L/2} f(x)dx = \int_{-a+\delta}^{a-\delta} f(x)dx$. Using the Stirling approximation, one can prove that the right-hand side of the inequality (\ref{ineqn:ch2.2}) approaches zero as $N$ approaches infinity. So we are done. This is the main idea of our proof, the full proof can be found in Appendix. 
\end{proof}

Proposition \ref{prop:glimit} suggests Equation (\ref{eqn:ch2_partial_entrain_thm_1}), a probabilistic lower bound on the number of oscillators in a partially entrained cluster. On the other hand, the probabilistic upper bound, given by Equation (\ref{eqn:ch2_partial_entrain_thm_2}) in Theorem \ref{thm:Main2}, is implied by the central limit theorem. We formalize it in the following proposition.

\begin{prop}
  Consider the finite $N$ Kuramoto model (\ref{eqn:Kuramoto})
  where the frequencies $\omega_i$ are independent and identically
  distributed according to a distribution with a density $f(\omega)$ that
  is symmetric and monomodal, with the unique  maximum of $f$ occuring
  at $\omega=0$. Then the
  probability that there is any partially entrained cluster
  containing more than
  \[
    N \int_{-\gamma}^\gamma f(\varpi)d\varpi + O(N^{\frac{1}{2}+\epsilon})
  \]
  tends to zero as $N \rightarrow \infty$. 
\end{prop}

\begin{proof}[Sketch of proof]
The proof of this is straightforward and similar to previous
arguments, so we just give the broad strokes. The basic observation is that from the usual
${\ell}_1/ {\ell}_\infty$ estimate we have that a subset of
oscillators cannot be partially entrained if
\[
  \omega_{max}-\omega_{min} \geq 2 \gamma.
\]
By the usual central limit theorem arguments the number of $\omega_i$
lying in an interval $I$ is, for $N$ large, approximately $\int_I
f(\omega)$. We would like to guarantee that (with high probability) there is no interval of length $I$ containing substantially more 
frequencies than that. Since $f$ is symmetric and monomodal the interval of
length $|I|=2\gamma $ which maximizes  $\int_I f(\omega)$ is the symmetric one,
so the largest cluster will, with high probability, have no more than
$\int_{-\gamma}^\gamma f(\omega) d\omega$.
\end{proof}

\begin{remark}
It is worth comparing this with the minimum cluster size guaranteed by Theorem \ref{thm:Main2}. The condition $g_\infty(\rho) \leq \gamma \tilde h(\rho)$
defines the largest guaranteed cluster size $1-\rho^*$ as a somewhat
complicated implicit function of the coupling strength $\gamma$, but
this simplifies greatly in the limit of large coupling strength
$\gamma$.
In the limit $\gamma \gg 1$ we have that $\rho \ll 1$ and the partial
synchronization condition becomes $g_\infty(\rho) \leq \gamma \tilde h(0) =
\gamma (2/3)^{\frac32}.$ Thus the theorem guarantees a partially
locked cluster of size at least
\[
 N_{\text{\rm cluster}} \gtrsim \int_{-(\frac23)^{\frac32}
    \frac{\gamma}{2}}^{(\frac23)^{\frac32} \frac{\gamma}{2}}f(\omega) d\omega
  \]
for large $\gamma$. 
\end{remark}

\section{Numerical Examples}
In this section, we give two examples to support Theorem \ref{thm:Main2}. In the first example, we consider oscillators with Gaussian distributed natural frequencies. In the second example, we consider oscillators with Cauchy distributed natural frequencies.

\begin{example}

For the case of Gaussian distributed natural frequencies $\omega_i$ the function  $g_{\infty}(\rho)$ is the inverse function to the error function:
\[
g_{\infty}(\rho) = 2\sqrt{2} \erf^{-1}(1-\rho).
\] 
Numerical calculations show that in the thermodynamic limit the minimum coupling in order to
guarantee the existence of partially entrained states is $\gamma^*
\approx 8.0027\sigma $, where $\sigma$ is the variance of the Gaussian
distribution (it is clear from scaling that the critical coupling
strength should be proportional to the variance).  For this critical
value of $\gamma$  we 
have $g_{\infty}(\rho) = \gamma^* \tilde{h}(\rho)$
at $\rho \approx .0901$.  Thus for a Gaussian distribution of
frequencies the theorem guarantees the existence of a
partially synchronized cluster containing all but about $9\%$ of the oscillators.

We illustrate Proposition \ref{prop:glimit}, with $N=10000$
oscillators with coupling strength $\gamma=10$ and suppose the natural
frequencies follow standard Gaussian distribution
$\mathcal{N}(0,1)$. In Figure \ref{fig:gfunc} we plot the function 
\[
g(\rho) = \min_{S\subset \Omega ~~|S|=(1-\rho)N} \max_{i,j \in S} |\omega_i-\omega_j|,
\]
 the function $g_{\infty}(\rho)$ and the curves $g_\infty(\rho)\pm
\frac{1}{\sqrt{N}}$. One can see that, as expected, the  actual
curve typically lies within $O(N^{-\frac12})$ of the limiting curve. 

We note at this point that it is difficult to see a sharp distinction
between partial phase-locking regime and the full phase-locking regime
for Gaussian distributed random variables in numerical
simulations. The reason for this is clear: partial phase-locking takes
place in the mean-field scaling
\[
\frac{d\theta_i}{dt} = \omega_i + \frac{\gamma}{N} \sum \sin(\theta_j-\theta_i)
\] 
while for Gaussian distributed frequencies full phase-locking takes
place in the slightly more strongly coupled scaling 
\[
\frac{d\theta_i}{dt} = \omega_i + \frac{\gamma\sqrt{\log{N}}}{N} \sum \sin(\theta_j-\theta_i).
\] 
In order to get a clean separation of scales one would like
$\sqrt{\log(N)}\gg 1$, which is numerically challenging. As an example
choosing $\sqrt{2\log{N}}\geq 8$ would guarantee that
(by the results of Bronski, DeVille and Park) that full phase locking
does not occur and (by the above) that partial phase-locking does
occur.  This would require an $N$ in the range 
$N\gtrsim 10^{14}$, which is not numerically feasible. The partial
phase-locking behaviour is much easier to observe for distributions
with broader tails. This motivates our next example, that of Cauchy
distributed frequencies. 

\begin{figure}[H]
    \centering 
    \includegraphics[width=.85\linewidth]{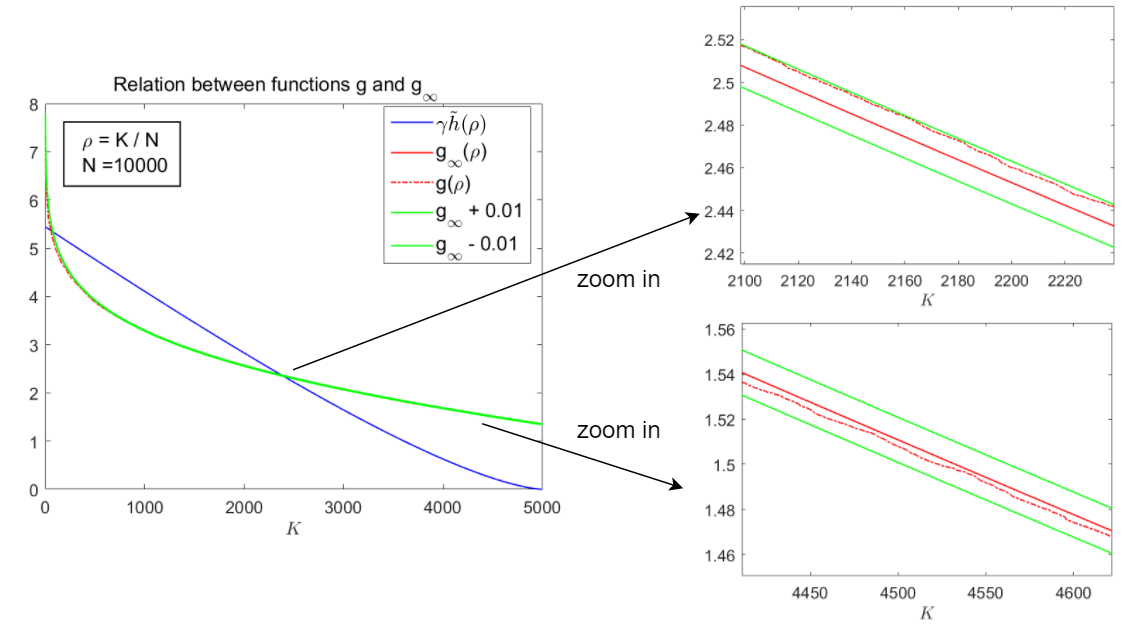}
    \caption{Comparison between functions $g$ and $g_{\infty}$ for
      Gaussian distribution}
\label{fig:gfunc}
\end{figure}

\end{example}

\begin{example}
For the case of Cauchy distributed natural frequencies $\omega_i$, their pdf and cdf are as follows:
\begin{align}
&f(\varpi;k,\lambda) = \frac{1}{k \pi (1+(\frac{\varpi-\lambda}{k})^2)},\\
&F(\varpi;k,\lambda) = \frac{1}{\pi} \arctan(\frac{\varpi-\lambda}{k}) + \frac{1}{2},
\end{align}
where $k$ is the scale parameter, and $\lambda$ is the location parameter, specifying the location of the peak of the distribution. We consider the case when $\lambda=0$.
The function  $g_{\infty}(\rho)$ is the inverse function to the cumulative distribution (sometimes called the quantile function):
\[
g_{\infty}(\rho) = 2\tan(\frac{\pi}{2} (1-\rho)).
\] 
Numerical calculations show that in the thermodynamic limit the minimum coupling in order to
guarantee the existence of partially entrained states is $\gamma^*
\approx 21.4950k $, where $k$ is the Cauchy scale parameter. It is clear from scaling that the critical coupling
strength should be proportional to the scale parameter $k$.  For this critical
value of $\gamma$  we 
have $g_{\infty}(\rho) = \gamma^* \tilde{h}(\rho)$
at $\rho \approx .2258$.  Thus for a Cauchy distribution of
frequencies the theorem guarantees the existence of a
partially synchronized cluster containing all but about $22.58\%$ of the oscillators for . $\gamma^*
\approx 21.4950k .$

As a numerical illustration of  Proposition \ref{prop:glimit}, we consider $N=10000$ oscillators with coupling strength $\gamma=50$ and suppose the natural frequencies follow Cauchy distribution with $k=1, \lambda=0$. We have that  $|g(\rho)- g_{\infty}(\rho)| = o( N^{-\frac12+\epsilon})$. The graphs of $g$ and $g_{\infty}$ are shown in Figure \ref{fig:ginfinity}, along with the curves $g_\infty \pm \frac{1}{\sqrt{N}} = g_\infty \pm 0.01$. As is clear from the figure we see the typical central limit type convergence of $g(\rho)$ to $g_\infty(\rho)$. 
 \begin{figure}[H]
    \centering 
    \includegraphics[width=.95\linewidth]{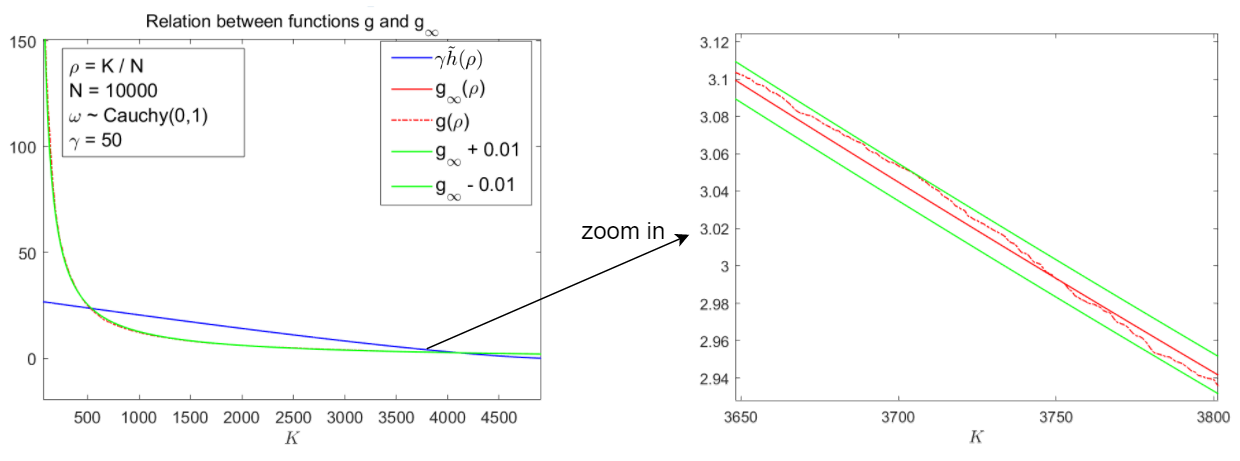}
    \caption{Comparison between functions $g$ and $g_{\infty}$ for Cauchy distribution}
    \label{fig:ginfinity}
\end{figure}

Next, we present a simulation to illustrate Theorem \ref{thm:Main2}.  In this simulation we take the Cauchy scale parameter to be  $k=1$ and the location parameter $\lambda=0$. We take $N=500$ oscillators, $\omega_1 = 0$ and $\omega_i \sim f(x; 1, 0)$ for $ i=2,...,N$, then direct calculation gives $\gamma^* = 21.4950$. Our numerical criteria for determining if an oscillator is part of the entrained cluster is as follows. We assume that oscillator number 1, which has zero frequency, is part of any entrained cluster.  Define $\Phi_{1i} = \theta_i(\frac{T}{2}) - \theta_1(\frac{T}{2})$, $\Phi_{2i} = \theta_i(T) - \theta_1(T)$ and $\Psi_i = (\Phi_{2i} - \Phi_{1i}) \times \frac{2}{T}$. Then we have
$
    \Psi_i \to \omega_{i\infty} - \omega_{1\infty} \ \text{as}\ T \to \infty,
$
and thus, $\Psi_i$ approaches zero if $\theta_i$ is locked with $\theta_1$. Now, define a relative frequency difference: 
$d = 10^{-5}\times (\max_i(\Psi_i) - \min_i(\Psi_i))$, where $10^{-5}$
is a tolerance that  we choose to classify phase-locked oscillators. If $\Psi_i \le d$, we regard $\theta_i$ as the oscillator that locks with $\theta_1$. To see the effect of $\gamma$ on the partial entrainment, we vary $\gamma$ from 1 to 25, and for each $\gamma$, use 5 samples of $\omega_i$ to solve Equation (\ref{eqn:Kuramoto}) numerically up to time $T=500$ with a time step $dt=0.1$. Then we compute the average number of oscillators in the largest cluster with frequency difference less than $d$, i.e, $\Psi_i \le d$, over the 5 simulations. The histogram graphs of the amount of oscillators corresponding to $\gamma=5$ and $\gamma=25$ are drawn separately in Figure \ref{fig:hist1}, where the $x$-axis is the frequency difference $\Psi_i$ and the $y$-axis is the average number of oscillators satisfying $\Psi_i \in (x-\frac{d}{2}, x+\frac{d}{2})$. The graphs show, as we expected, the size of the largest cluster of phase-entrained oscillators is larger for $\gamma=25$ than which of $\gamma=5$.

\begin{figure}[H]
    \begin{minipage}[b]{0.5\textwidth}
       \centering 
       \includegraphics[width=2.4in, height=1.8in]{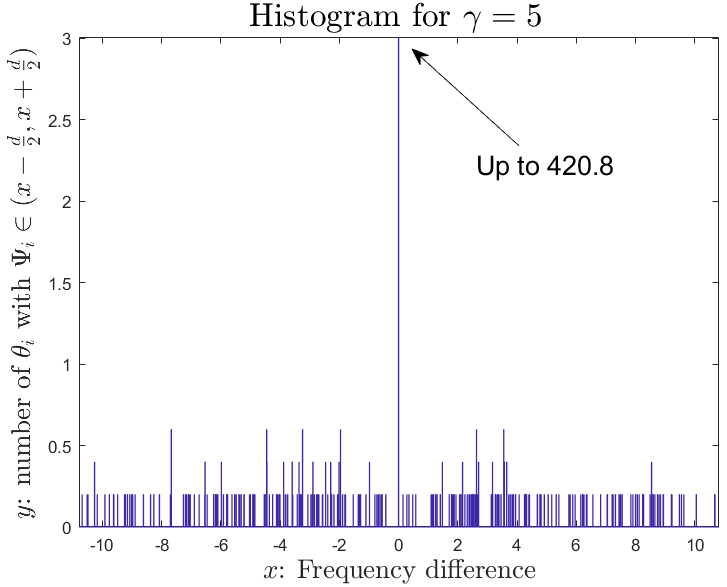}
    \end{minipage}%
    \begin{minipage}[b]{0.48\textwidth}
        \centering 
        \includegraphics[width=2.4in, height=1.8in]{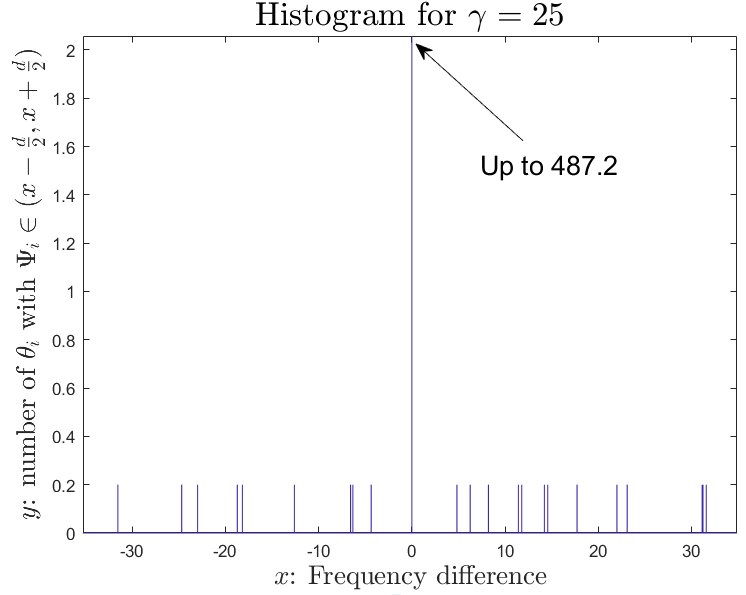}
    \end{minipage}
    \caption{Histogram graphs of average number of entrained oscillators: left: $\gamma=5$; right: $\gamma=25$}
    \label{fig:hist1}
\end{figure}

In the next experiment we define three percentages  $P_{\text{numeric}}$, $P_{\text{lower}}$ and $P_{\text{upper}}$ and make a careful comparison among them as a function of the coupling strength $\gamma$.
Firstly let $P_{\text{numeric}}$ denote the average percentage of oscillators in the largest phase-entrained cluster over the 5 simulations. For instance, the realizations in the right graph in Figure \ref{fig:hist1} give $P_{\text{numeric}} = \frac{487.2}{500} \approx 97\%$ when $\gamma=25$.
Secondly, we note the well-known $\ell_1-\ell_\infty$ estimate: namely that if we have
\begin{align}
|\omega_i - \omega_j| > 2\gamma \ge \frac{2\gamma}{N}|\sum_{j=1}^{N} \sin(\theta_j-\theta_i)|,
\end{align}
then the $ith$ oscillator and $jth$ oscillator will never synchronize. Thus, by the law of large number, the percentage of oscillators that lock together must be (with high probability) less than $\int_{-\gamma}^{\gamma} f(x;1,0) dx + o(1)$. We let $P_{\text{upper}}$ denote this percentage, i.e., $P_{\text{upper}} = \int_{-\gamma}^{\gamma} f(x;1,0) dx$. Finally, according to Theorem \ref{thm:Main2}, we know that as $\gamma > \gamma^* = 21.4950$,  there are at least $n = (1-\rho_{\min})\times N$ oscillators locking together, where $\rho_{\min}$ is defined as the $\rho$-coordinate of the first intersection point of $g_{\infty}$ and $\gamma\tilde{h}$. Let $P_{\text{lower}}$ denote the percentage of oscillators in the largest phase-entrained cluster derived from this theorem, i.e., $P_{\text{lower}} = 1 -\rho_{\text{min}}$.

Obviously, we have the following inequality
\begin{equation}
    P_{\text{lower}} \le P_{\text{numeric}} \le P_{\text{upper}}.%
    \label{ineqn:3Ps}
\end{equation}
As a numerical  check of inequality (\ref{ineqn:3Ps}), we consider a sequence of values of the coupling strength $\gamma$. For each value of $\gamma$ we plot $P_{\text{numeric}}$, the percentage of oscillators in the largest entrained cluster, as well as $P_{\text{upper}}$ and $P_{\text{lower}}$. Note that when $\gamma<\gamma^* = 21.4950$, functions $g_{\infty}$ and 
$\gamma\tilde{h}$ have no intersections, so our theorem cannot guarantee any cluster of phase-entrained oscillators. Therefore, $P_{\text{lower}}=0$ when $\gamma<\gamma^* = 21.4950$, as seen in Figure \ref{fig:3P_cauchy}. It is clear that, at least for the range of $\gamma$ considered the 
upper bound from the $\ell_1-\ell_\infty$ estimate and the law of large numbers is actually a very good approximation to the observed number of entrained oscillators. 

\begin{figure}[H]
    \centering 
    \includegraphics[width=.65\linewidth]{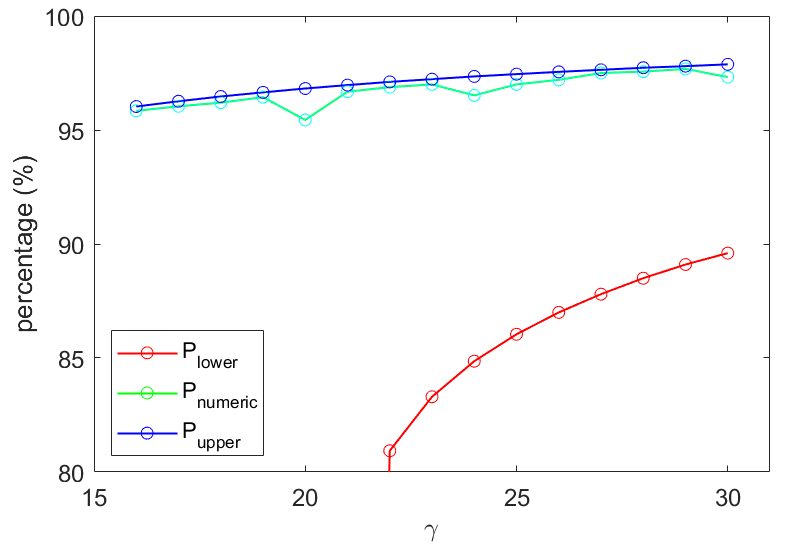}
    \caption{Largest cluster of partial synchronization for Cauchy distribution}
    \label{fig:3P_cauchy}
\end{figure}

It is interesting to consider the asymptotic  percentage of entrained oscillators for large coupling strength $\gamma$. Note that as $\gamma$ grows large, $\rho_{\text{min}}$ tends to approach zero, and thus, $ \tilde{h}(\rho)$ approaches $(\frac{2}{3})^{\frac{3}{2}}$. From Equation (\ref{eqn:gh}), we have $$g_{\infty}(\rho)  \to (\frac{2}{3})^{\frac{3}{2}}\gamma\approx 0.544\gamma \ \text{as} \ \rho \to 0.$$
Using the definition of $g_{\infty}$ as given by (\ref{eqn:ginfty}), it is easy to compute that $\rho_{\text{min}} = 2\int_{(\frac{2}{3})^{\frac{3}{2}}\frac{\gamma}{2}}^{\infty} f(x) dx \approx 2\int_{0.272\gamma}^{\infty} f(x) dx$. Thus, when $\gamma$ is large,
\begin{equation}
    P_{\text{lower}} \sim 1 - 2\int_{0.272\gamma}^{\infty} f(x) dx .
    \label{Plower}
\end{equation}
Denote the right-hand side as $P_{\text{lower asym}}$, i.e., $P_{\text{lower asym}} = 1 - 2\int_{0.272\gamma}^{\infty} f(x) dx$. Then for the Cauchy distribution, $(1-P_{\text{lower asym}})\sim \frac{1}{\gamma}$, i.e., the percentage of unlocked oscillators is inversely proportional to the coupling strength when the strength is large. On the other hand, for $P_{\text{upper}}$, by its definition, we have for any $\gamma>0$,
\begin{equation}
    P_{\text{upper}} = 1 - 2\int_{\gamma}^{\infty} f(x) dx .
    \label{eqn:Pupper}
\end{equation}

The order parameter $r$, defined by
\begin{equation}
    r(t) = \mid \frac{1}{N}\sum\limits_{j=1}^N e^{i\theta_j(t)} \mid,
\end{equation}
is a widely used proxy for synchronization. It is worthwhile to plot the evolution of the order parameter as a function of time for some different values of the coupling strength $\gamma$. Specifically we choose 
$\gamma = \gamma^*/2, \gamma^*, 2\gamma^*$, where $\gamma^*$ is the minimum coupling strength required by the theorem in order to guarantee the existence of a partially entrained state. As one can see from Figure \ref{fig:order_param}
we see the prder parameter $r(t)$ oscillate around a non-zero mean for values of $\gamma$ substantially below the $\gamma^*$ required by the theorem. 

\begin{figure}[H]
    \centering 
    \includegraphics[width=.65\linewidth]{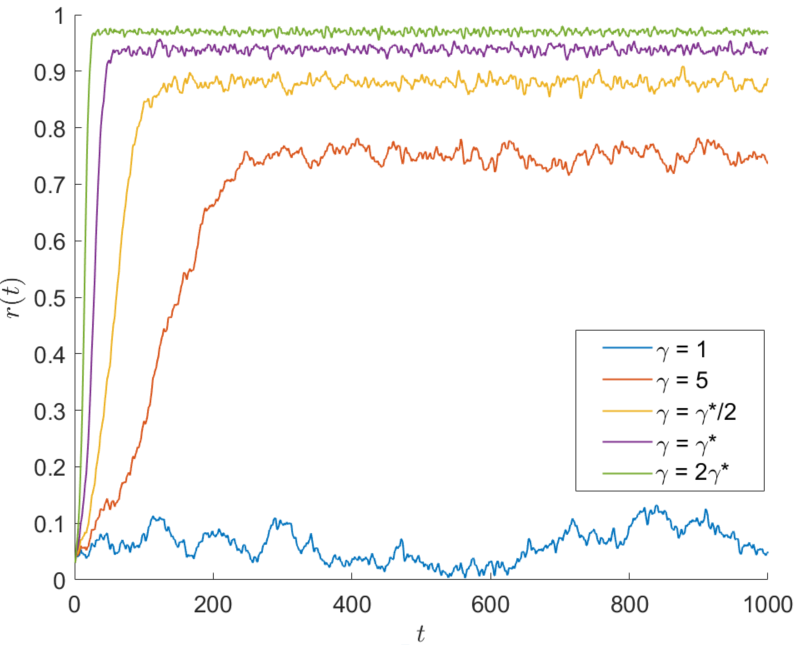}
    \caption{Order parameter $r(t)$ for different coupling strength}
    \label{fig:order_param}
\end{figure}

\end{example}

\section{Conclusions}

In this paper, we derived an explicit analytical expression for a sufficient condition on the coupling strength $\gamma$ to achieve partial phase-locking (entrainment) in the classical finite-N Kuramoto 
model (\ref{eqn:Kuramoto}) for any arbitrary monomodal distribution of the natural frequencies.  We also derived explicit upper and lower bounds on the percentage of entrained 
oscillators, again as a function of the coupling strength. This result can be veiwed as an extension of the result of  F.D\"orfler and F.Bullo \cite{Drfler2011} on full 
phase-locking to the case of partial phase locking. The requirement that the distribution of frequencies be monomodal is interesting, in that other authors have identified a change in the nature of 
the bifurcation when one moves from mono-modal distributions to bimodal or trimodal. In the work of Acebr\'on, Perales and Spigler\cite{APS}, for instance, the authors identify a change in the nature of the bifurcation, 
from subcritical to supercritical, as one moves from monomodal to multimodal distributions. 

While the scaling of the result is optimal -- it holds in the usual mean field scaling whereas, for instance, full phase-locking requires a slightly stronger coupling than the mean field coupling -- the constants are almost certainly 
not optimal and could likely be improved. It is interesting, in fact, that the numerical experiments suggest that the size of the largest entrained cluster is well-predicted by the upper bound given by the law of large numbers. It would be interesting 
to see if one could derive a lower bound that is closer to the current upper bound.

\newpage
\bibliographystyle{plain}
\bibliography{main}

\newpage
\begin{appendices}

\section{Proof of Proposition \ref{prop:ball}}

\begin{proof}
Suppose $J$ is the Jacobian matrix of (\ref{eqn:Kuramoto}) at $\theta^*$, $\lambda_1, \lambda_2,..., \lambda_N$ are $N$ eigenvalues of $J$ and $v_1, v_2,...,v_N$ are the corresponding eigenvectors. Since $\theta^*$ is a stable fixed point, by definition, $\lambda_N \le \lambda_{N-1} \le ... \le \lambda_2 < \lambda_1 = 0$. And clearly, $v_1 = \hat{\textbf{1}} = (1,1,...,1)$. Let $V=Ker(J)= \text{span} \{\hat{\textbf{1}}\}$ and $W=\text{span} \{v_2, v_3,...,v_N\}$, then $V\oplus W = \mathbb{R}^N$.

Now, consider any steady solution of Equation (\ref{eqn:pert_Kura}) that is close to $\theta^*$, i.e., consider $\theta = \theta^* + \tilde{\theta}$ where $||\tilde{\theta}||$ is small. Then we have

\begin{align}
\dot{\theta_i} = \dot{\tilde{\theta_i}} 
&= \omega_i + \frac{\gamma}{N}\sum_{j} \sin(\theta_j^* - \theta_i^* + (\tilde{\theta_j} - \tilde{\theta_i})) + \epsilon f_i(\theta,t) \\
&= \omega_i + \frac{\gamma}{N}(\sum_{j} \sin(\theta_j^* - \theta_i^*) + \sum_{j} \cos(\theta_j^* - \theta_i^*)(\tilde{\theta_j} - \tilde{\theta_i}) - \sum_{j} \sin(\xi_{i,j})\frac{(\tilde{\theta_j} - \tilde{\theta_i})^2}{2}) \\
& \quad + \epsilon f_i(\theta,t) \qquad \text{where} \ \xi_{ij} \ \text{is between} \ (\theta_j^* - \theta_i^*) \ \text{and} \ (\theta_j - \theta_i) \hspace{30mm}\\
&= \frac{\gamma}{N} \sum_{j} \cos(\theta_j^* - \theta_i^*)(\tilde{\theta_j} - \tilde{\theta_i}) - \frac{\gamma}{N}\sum_{j} \sin(\xi_{i,j})\frac{(\tilde{\theta_j} - \tilde{\theta_i})^2}{2} + \epsilon f_i(\theta,t) \\
&= (J\tilde{\theta})_i - \frac{\gamma}{N}\sum_{j} \sin(\xi_{i,j})\frac{(\tilde{\theta_j} - \tilde{\theta_i})^2}{2} + \epsilon f_i(\theta,t), 
\end{align}
where $(J\tilde{\theta})_i$ refers to the $i$th row of the matrix $J\tilde{\theta}$.\\

By our definition of semi-norm (\ref{def:semi_norm}), $||\theta||^2 = ||\theta||^2_\Omega =\frac{1}{N} \sum\limits_{1 \le i \le j \le N} (\theta_i - \theta_j)^2 = \frac{1}{N} \theta^T M \theta$, where $M=$
\begin{equation*} 
\left(                 
  \begin{array} {cccccc}  
    N-1 & -1 & -1 & ... & -1\\ 
    -1 & N-1 & -1 & ... & -1\\ 
    \quad & \quad & ......  \\
    -1 & -1 & -1 & ... & N-1\\
  \end{array}
\right).                 
\end{equation*}

Notice that $M$ has an eigenvalue 0 with multiplicity 1 and an eigenvalue $N$ with multiplicity $N-1$, so $M$ is positive semi-definite. By computing the derivative of this semi-norm for $\tilde{\theta} \in W$, we have
  
\begin{align*} 
\frac{d}{dt}\Vert\tilde\theta\Vert^2 &= \frac{1}{N}\frac{d}{dt}
                                     \tilde{\theta}^TM \tilde{\theta}
                                     = \frac{2}{N} \tilde \theta^T M\dot\tilde\theta  \\
& \le \frac{2}{N} \tilde\theta^T M J\tilde\theta +\frac{\gamma}{N^2}\sum_{i,j}| \sin(\xi_{i,j})[(\tilde\theta_i - \tilde\theta_j)^2 \sum_{k}\tilde\theta_i - \tilde\theta_k)]| + \frac{2}{N}\epsilon \tilde\theta^T M f  \\
& \le \frac{2}{N}\tilde\theta^T M J\tilde\theta + \frac{\gamma}{N^2}\sum_{i,j}[(\tilde\theta_i - \tilde\theta_j)^2 \sum_{k}|\tilde\theta_i - \tilde\theta_k|] + \frac{2\epsilon}{N}||\tilde{\theta}^TM||\cdot||f|| \\
& \le \lambda_2 \frac{2}{N}\tilde\theta^T M \tilde\theta + \frac{\gamma}{N^2}\sum_{i,j}(\tilde\theta_i - \tilde\theta_j)^2 (\sum_{k}(\tilde\theta_i - \tilde\theta_k)^2)^{1/2} N^{1/2} + \frac{2\epsilon }{N}  (N\tilde{\theta}^TM\tilde{\theta})^{1/2}\cdot(N^{1/2}C)\\
& \le 2\lambda_2 \Vert\tilde\theta\Vert^2 + \gamma \Vert\tilde\theta\Vert^3 + 2 \epsilon C N^{1/2} \Vert\tilde\theta\Vert .  
\end{align*}

For $\tilde\theta \in V$, $\Vert\tilde\theta\Vert = 0$. In this case, since $M\cdot J$ is negative semi-definite, we still have above inequality. Thus for any small $\tilde{\theta} \in \Re^N$, we have
\begin{align*}
\frac{d}{dt}{||\tilde{\theta}||^2} \le 2\lambda_2 ||\tilde{\theta}||^2 + \gamma ||\tilde{\theta}||^3 + 2\epsilon C N^{1/2} ||\tilde{\theta}||.
\end{align*}

To find the basin of attraction, it suffices to find the domain of $||\tilde{\theta}||$ such that
\begin{equation} 
2\lambda_2 ||\tilde{\theta}||^2 + \gamma ||\tilde{\theta}||^3 + 2 \epsilon C N^{1/2} ||\tilde{\theta}|| < 0,
\end{equation}
which will be satisfied if 
\begin{equation}
 \left\{
   \begin{array}{c}
   2\epsilon C N^{1/2} ||\tilde{\theta}|| < c_1|\lambda_2| ||\tilde{\theta}||^2  \\
  \gamma||\tilde{\theta}||^3 < c_2|\lambda_2| ||\tilde{\theta}||^2, \\
   \end{array}
 \right.
\end{equation}
where $c_1>0, c_2>0$ and $c_1+c_2 \le 2$. So we need
\begin{equation}
\frac{2\epsilon CN^{1/2}}{c_1|\lambda_2|} < ||\tilde{\theta}|| < \frac{c_2|\lambda_2|}{\gamma}.
\label{ineqn:basin_attract}
\end{equation}

It's clear to see (\ref{ineqn:basin_attract}) makes sense only when  $\epsilon < \frac{c_1c_2|\lambda_2|^2}{2CN^{1/2}\gamma}$. Since $c_1c_1 \le (\frac{c_1+c_2}{2})^2 \le 1$, the loosest bound on $\epsilon$ is $\frac{|\lambda_2|^2}{2CN^{1/2}\gamma}$, when $c_1=c_2=1$. 

Let $r(\epsilon) = 2\epsilon CN^{1/2}/|\lambda_2|$ and $R =
|\lambda_2|/\gamma$, then by Gronwall's inequality
\cite{Khalil1993}, the semi-norm of $\tilde{\theta}$ is
exponentially decreasing when $\tilde{\theta}$ is in the annulus of
radii $r(\epsilon)$ and $R$, and then stays in the ball of radius
$r(\epsilon)$ forever. So statements (1) and (2) in Proposition 3.2 were proved.
\end{proof}

\section{Proof of Proposition \ref{prop:glimit}}
\begin{proof}
The goal is to prove Equation (\ref{eqn:glimit}):
$$\lim_{N \rightarrow \infty}{\mathbb P} (|g(K,N) - g_{\infty}(\frac{K}{N})| \leq N^{-\frac12+\epsilon}) = 1,$$
in other words, we need
\begin{align}
    & \lim_{N \rightarrow \infty}{\mathbb P} (g \leq  g_{\infty} + N^{-\frac12+\epsilon}) = 1,
    \label{eqn:g_upper_bound}\\
    & \lim_{N \rightarrow \infty}{\mathbb P} (g \geq g_{\infty} - N^{-\frac12+\epsilon}) = 1.
    \label{eqn:g_lower_bound}
\end{align}
For simplicity, define $a = F^{-1}(1-\frac{\rho}{2})$ so that we have $g_\infty = 2a$ and define $\delta=\frac{1}{2}N^{-\frac{1}{2}+\epsilon}$, then Equations (\ref{eqn:g_upper_bound}) and (\ref{eqn:g_lower_bound}) can be rewritten as 
\begin{align}
    & \lim_{N \rightarrow \infty}{\mathbb P} (g \leq  2(a+\delta)) = 1,
    \label{eqn2:g_upper_bound}\\
    & \lim_{N \rightarrow \infty}{\mathbb P} (g \geq 2(a-\delta)) = 1.
    \label{eqn2:g_lower_bound}
\end{align}
Let's prove Equation (\ref{eqn2:g_upper_bound}) first. In fact, we will show $\mathbb{P}(g\leq 2a)$ tends to one as $N\to\infty$. Define $X_i$ =
\begin{equation}
\begin{cases}
1, & \text{if} \ \omega_i \in [-a,a] \\
0, & \text{if} \ \omega i \notin [-a,a].\\
\end{cases}
\end{equation}
Then $X_i's$ are i.i.d random variables since $\omega_i's$ are i.i.d random variables. Let $X=X_1+X_2+...+X_N$, then $X$ represents the number of $\omega_i$ such that $\omega_i \in [-a,a]$. By  strong law of large number theorem, $\frac{X}{N}$ converges to ${\mathbb E}(X_i)$ almost surely, i.e., ${\mathbb P}(\underset{N \to \infty} {\lim } \frac{X}{N} = \int_{-a}^{a} f(x) dx) = 1$. Notice that $\int_{-a}^{a} f(x) dx =  1 - \rho$, so we have ${\mathbb P}(\underset{N \to \infty} {\lim } \frac{X}{N} = 1- \rho) = 1$. Moreover, we know $g(\rho)\le 2a$ if $X=(1-\rho)N$ by the definition of the function $g$. Therefore, ${\mathbb P}(g(\rho) \le 2a) = 1$ as $N \rightarrow \infty$. Equation (\ref{eqn2:g_upper_bound}) has been proved.

The other direction Equation (\ref{eqn2:g_lower_bound}) is less trivial to prove. Intuitively, we want to show that with high probability no intervals with length $g_\infty-2\delta$ contain more than $(1-\rho)N$ points. To show this, we need to firstly make two important observations. First, notice that if no intervals of Length $L$ with $\omega_k$ at an endpoint contain more than $m$ points then no any other interval does. So we can only focus on $N$ intervals $\{[\omega_i, \omega_i+L]: i=1,2,...,N\}$. Second, the interval centered at zero maximizes the probability that a point lies in the interval, i.e., $I=[-L/2, L/2]$ gives the largest $\mathbb{P}(x\in I)$ among all intervals of length $L$. The proof follows from the fact that for $\mu=\int_a^{a+L}f(x)dx$, its derivative $\frac{d\mu}{da}=f(a+L)-f(a)$ is zero when $a=-\frac{L}{2}$. As a result, the probability that the interval of length $L$ with $\omega_k$ at an endpoint contains more than $m$ points is less than the probability that $[-L/2, L/2]$ contains more than $m$ points. Now, fix $L=2(a-\delta)$ where $\delta$ is defined at the beginning of the proof. Define $A_k$ as the event that interval $[\omega_k, \omega_k+L]$ containing more than $(1-\rho)N$ points, $A$ as the event that there exists an interval with length $L$ containing more than $(1-\rho)N$ points, and $B$ as the event that $[-L/2,L/2]$ contains more than $(1-\rho)N$ points. Clearly, our goal is to prove
\begin{equation}
    \mathbb{P}(A) \to 0\ \text{as}\ N\to\infty.
\end{equation}

Due to the above two observations and the union bound, we have
\begin{equation}
   \mathbb{P}(A) = \mathbb{P}(\cup_{k=1}^N A_k) \leq \sum_{k=1}^N \mathbb{P}(A_k) \le N\cdot\mathbb{P}(B).
\end{equation}
Note that
\begin{equation}
    \mathbb{P}(B) =  N\sum\limits_{M=\ceil{(1-\rho)N}}^N \binom{N}{M} p^M(1-p)^{N-M},
    \label{eqn:prob_B}
\end{equation}
where $1-\rho = \int_{-a}^a f(x)dx$ and $p=\int_{-L/2}^{L/2} f(x)dx = \int_{-a+\delta}^{a-\delta} f(x)dx$. we denote the right-hand side of Equation (\ref{eqn:prob_B}) as $R(\rho, p, N)$, then it is sufficient to show 
\begin{equation}
    R(\rho, p, N) \to 0 \ \text{as}\ N\to\infty.
\end{equation}
Define $\tau_{N,M}:=\binom{N}{M}p^M(1-p)^{N-M}=\frac{N!}{M!(N-M)!}p^M(1-p)^{N-M}$. Then $$\log(\tau_{N,M}) = \log(N!)-\log(M!)-\log(N-M)!+M\log(p)+(N-M)\log(1-p).$$
For large $N$, using Stirling's approximation: $\log(N!) \approx N\log(N)-N+\frac{1}{2}\log(2\pi N)$, we have
\begin{align*}
    \log(\tau_{N,M}) &\approx N\log(N) - N + \frac{1}{2}\log(2\pi N) - M\log(M) + M - \frac{1}{2}\log(2\pi M) \\
     &\quad -(N-M)\log(N-M) + (N-M) - \frac{1}{2}\log(2\pi (N-M))\\
     &\quad + M\log(p) - (N-M)\log(1-p)\\
    &= N\log(N) - M\log(N) - M\log(\frac{M}{N}) - (N-M)\log(N)\\
    &\quad - (N-M)\log(\frac{N-M}{N}) + M\log(p) - (N-M)\log(1-p)\\
    &\quad + \frac{1}{2}\log\left(\frac{N}{2\pi M(N-M)}\right)\\
    &= N\left(-\frac{M}{N}\log(\frac{M}{N}) - (1-\frac{M}{N})\log(1-\frac{M}{N}) + \frac{M}{N}\log(p) + (1-\frac{M}{N})\log(1-p)\right)\\
    &\quad + \frac{1}{2}\log\left(\frac{N}{2\pi M(N-M)}\right)\\
    & = N\left(-x\log(x) - (1-x)\log(1-x) + x\log(p) + (1-x)\log(1-p)\right)\\
    &\quad + \frac{1}{2}\log\left(\frac{N}{2\pi x(1-x)}\right) - \log(N) \quad \text{by setting}\ x=\frac{M}{N}.
\end{align*}
Let 
\begin{equation}
    \phi(x) := -x\log(x) - (1-x)\log(1-x) + x\log(p) + (1-x)\log(1-p),
\end{equation}
then 
\begin{equation}
    \phi'(x) = \log(\frac{p}{1-p}) - \log(\frac{x}{1-x})\ \text{ and} \ \phi''(x) = \frac{-1}{x(1-x)}<0.
\end{equation}
So $\phi$ reaches the largest when $x=p$. And thus, when $N$ is large, the maximum of $\tau_{N,M}$ occurs when $x=\frac{M}{N}=p$. In the neighborhood of the maximum: $x= p + y$, $\phi(x)\approx \frac{-1}{2p(1-p)}y^2$. So we have
\begin{equation}
   \log(\tau_{N,M}) = \frac{1}{2}\log(\frac{1}{2\pi p(1-p)N}) - \frac{N}{2p(1-p)}y^2 + O(y), 
\end{equation}
and thus 
\begin{equation}
    \tau_{N,M} \approx \frac{1}{N} \sqrt{\frac{N}{2\pi p(1-p)}} e^{-\frac{Ny^2}{2p(1-p)}+ O(y)} .
\end{equation}
Recall that $1-\rho = \int_{-a}^a f(x)dx$ and $p=\int_{-a+\delta}^{a-\delta} f(x)dx$ where $\delta=N^{-\frac{1}{2} +\epsilon}$, then $(1-\rho)-p \sim N^{-\frac{1}{2} +\epsilon}$. So for $M\ge \ceil{(1-\rho)N}$, we have $\frac{M}{N}-p \gtrsim N^{-\frac{1}{2} +\epsilon}$, i.e.,  $y \gtrsim N^{-\frac{1}{2} +\epsilon}$. On the other hand, $y\le 1-p<1$. Thus
\begin{equation}
    \tau_{N,M} \lesssim \frac{1}{N}\sqrt{\frac{N}{2\pi p(1-p)}} e^{-\frac{N^{2\epsilon}}{2p(1-p)}+O(1)}.
\end{equation}
So we have
\begin{align*}
    R(\rho, p, N) &\le N \cdot N \cdot \left(\frac{1}{N}\sqrt{\frac{N}{2\pi p(1-p)}} e^{-\frac{N^{2\epsilon}}{2p(1-p)}+O(1)}\right)\\
    & = N\sqrt{\frac{N}{2\pi p(1-p)}} e^{-\frac{N^{2\epsilon}}{2p(1-p)}+O(1)},
\end{align*}
which implies $R(\rho, p, N) \to 0$ as $N\to\infty$ for any positive $\epsilon$. The proof of Equation (\ref{eqn2:g_lower_bound}) is now complete.

With the two equations (\ref{eqn2:g_upper_bound}) and (\ref{eqn2:g_lower_bound}), we proved Proposition (\ref{prop:glimit}).

\end{proof}

\end{appendices}

\end{document}